\documentclass{acm_proc_article-sp}

\usepackage{amssymb}
\usepackage{amsfonts}
\usepackage{amsmath}
\usepackage{graphicx}
\usepackage[caption=false]{subfig}
\usepackage{multirow}
\usepackage{color}
\usepackage{comment}
\usepackage{algorithm}
\usepackage{algorithmic}
\usepackage{float}
\newfloat{algorithm}{t}{lop}

\def\P1{\mathcal{P}_1}

\newdef{theorem}{Theorem}
\newdef{proposition}{Proposition}
\newdef{corollary}{Corollary}
\definecolor{Gray}{rgb}{0.5,0.5,0.5}

\begin{document}
%

\title{Fast Generation of Large Scale Social Networks with Clustering}
\numberofauthors{1}

\author{
\alignauthor Joseph J. Pfeiffer III, Timothy La Fond, Sebastian Moreno, Jennifer Neville\\
       \affaddr{Purdue University}\\
       \affaddr{Department of Computer Science}\\
       \affaddr{West Lafayette, IN}\\
       \email{\{jpfeiffer, tlafond, smorenoa, neville\}@purdue.edu}
}

\maketitle

\begin{abstract}
A key challenge within the social network literature is the problem of network generation -- that is, how can we create synthetic networks that match characteristics traditionally found in most real world networks?  Important characteristics that are present in social networks include a power law degree distribution, small diameter and large amounts of clustering; however, most current network generators, such as the Chung Lu and Kronecker models, largely ignore the clustering present in a graph and choose to focus on preserving other network statistics, such as the power law distribution.  Models such as the exponential random graph model have a transitivity parameter, but are computationally difficult to learn, making scaling to large real world networks intractable.

In this work, we propose an extension to the Chung Lu random graph model, the Transitive Chung Lu (TCL) model, which incorporates the notion of a random transitive edge.  That is, with some probability it will choose to connect to a node exactly two hops away, having been introduced to a `friend of a friend'.  In all other cases it will follow the standard Chung Lu model, selecting a `random surfer' from anywhere in the graph according to the given invariant distribution.  We prove TCL's expected degree distribution is equal to the degree distribution of the original graph, while being able to capture the clustering present in the network.  The single parameter required by our model can be learned in seconds on graphs with millions of edges, while networks can be generated in time that is \emph{linear} in the number of edges.  We demonstrate the performance TCL on four real-world social networks, including an email dataset with hundreds of thousands of nodes and millions of edges, showing TCL generates graphs that match the degree distribution, clustering coefficients and hop plots of the original networks.

\end{abstract}


\pagebreak

\section{Introduction}
A challenging problem within the social network community is generating graphs which adhere to certain statistics.  Due to the prevalence of `small world' graphs such as Facebook and the Internet \cite{Watts_SmallWorld}, models which attempt to capture properties of small world graphs such as a power law degree distribution, small diameter and clustering greater than randomly present for the sparsity of the network have become a much-discussed topic in the field \cite{Leskovec:Kronecker, Robins_ERGM, ChungLu, Pinar, BarAlb99, Seshadhri:11}.  The first random graph model, the Erdos-Renyi model\cite{Erdos_RandomGraph}, proposed random connections between nodes in the graph where each edge is sampled independently; however, this model has a Binomial degree distribution, not power law, and generally lacks clustering when generating sparse networks.  As a result, multiple attempts have been made to develop algorithms that generate graphs with small world network properties.

Exponential Random Graph Models (ERGM) extend the Erdos-Renyi model to allow additional statistics of the graph as parameters \cite{Robins_ERGM}.  The typical approach is to model the network under the assumption of Markov independence throughout the graph -- edges are only dependent on other edges that share the same node .  Using this, ERGMs define an exponential family of models using various Markov statistics of the graph, allowing for the incorporation of a transitivity parameter, then maximize the likelihood of the parameters given the graph.  The algorithms for learning and generating ERGMs are resource intensive and intractable for application to networks of more than a few thousand nodes.

As a result, newer efforts make scaleability an explicit goal when constructing models and algorithms. Notable examples include the Chung-Lu Graph Model (CL) \cite{ChungLu} and the Kronecker Product Graph Model (KPGM) \cite{Leskovec:Kronecker}.  CL is also an extension of the Erdos-Renyi model, but rather than creating a summary statistic based on the degrees, it generates a new graph such that the expected degree distribution matches the given distribution exactly. 
In contrast, KPGM learns a 2x2 matrix of parameters and lays down edges according to the Kronecker product of the matrix to itself $\log n$ times.  For large graphs this algorithm can learn the parameters defined by the 2x2 matrix in hours and can generate large graphs in minutes.

With CL and KPGM we have scalable algorithms for learning and generating graphs with hundreds of thousands of nodes and millions of edges.  However, in order to achieve scalability, a power law degree distribution and small diameter, both models have made the decision to ignore \emph{clustering} in their generated graphs.  This is not an insignificant consequence, as a small world network is in part defined by the clustering of nodes \cite{Watts_SmallWorld}.  While ERGM can potentially learn networks with clustering, the complexity of the model makes it a poor prospect when considering learning and generating graphs with massive size.

In order to generate sparse networks which can accurately capture the degree distribution, small diameter \emph{and} clustering, we propose to extend the CL algorithm in multiple ways.  The first portion of this paper will show how the naive fast generation algorithm for the CL model is biased, and we develop a correction to this problem.  Next, we introduce a generalization to the CL model known as the \emph{Transitive Chung Lu} (TCL).  To do this, we observe that the CL model is a `random surfer' model, similar to the PageRank random walk algorithm\cite{PageRank}.  However, CL always chooses the random surfer and has no affinity for nodes along transitive edges.  In contrast, our TCL model will sometimes choose to follow these transitive edges and then close a triangle rather than selecting a random node according to the surfer.  The probability of randomly surfing versus closing a triangle is a single parameter in our model which we can learn in \emph{seconds} from an observed graph, compared to the hours required to learn KPGM.
In short, the contributions in our work can be summarized as follows:

\begin{itemize}
\item Introduction of a `random triangle' parameter to the CL model
\item A correction to the `edge collision' problem seen in naive fast CL model generation
\item Analysis showing TCL has an expected degree distribution equal to the original input network's degree distribution
\item A learning algorithm for TCLs that runs in seconds for graphs with millions of edges
\item A generation algorithm for TCLs which runs on the same order as naive fast CL, and faster than KPGM
\item Empirical demonstrations that show the graphs generated from TCL match the degree distribution, clustering coefficient and hop plots of the original graph better than fast CL or KPGM
\end{itemize}

In section \ref{relwork} we discuss in more depth the ERGM, KPGM and CL models, while in section \ref{invariantmc} we outline the basis for the CL model.   Next, we show the fast method used for generating graphs in section \ref{sec:fastcor}, and our correction to it.  In section \ref{sec:tcl} we introduce our modification to the CL model, proving the expected degree distribution and demonstrating how to learn the transitive probability, while in section \ref{sec:tc} we analyze the runtimes of our fast CL correction and TCL.  In section \ref{sec:experiments} we learn the parameter and generate graphs which closely match the original graphs.  We end in section \ref{sec:conclusions} with conclusions and future directions.

\section{Related Work}\label{relwork}
Recently there has been a great deal of work focused on the development of generative models for small world and scale-free graphs (e.g.,~\cite{frank:86,Watts_SmallWorld,BarAlb99,kumar:00,WasPat96,Leskovec:Kronecker,ChungLu}). As an example, the Chung Lu model is able to generate a network which has a provable expected degree distribution equal to the degree distribution of the original graph.  The CL model, like many, attempts to define a process which matches a subset of features observed in a network.

The importance of the clustering coefficient has been demonstrated by Watts and Strogatz \cite{Watts_SmallWorld}. In particular, they show that small world networks (including social networks) are characterized by a short path length and large clustering coefficient.  One recent algorithm (Seshadri et al \cite{Seshadhri:11}) matches these statistics by putting together nodes with similar degrees and generating Erdös-Rényi graphs for each group. The groups are then tied together.  However, this algorithm needs a parameter to be set manually to work. Existing models that can generate clustering in the network generally do not have a training algorithm.

One method that can model clustering and can learn the associated parameter is the Exponential Random Graph Model (ERGM) \cite{WasPat96}. ERGMs define a probability distribution over the set of possible graphs with a log-linear model that uses feature counts of local graph properties. However, these models are typically hard to train as each update of the Fisher scoring function takes $O(n^2)$.  With real-world networks numbering in the hundreds of thousands if not millions of nodes, this makes ERGMs impossible to fit.

Another method is the Kronecker product graph
model (KPGM), a scalable algorithm for learning models of large-scale networks that empirically preserves a wide range of global properties of interest, such as degree distributions, and path-length distributions~\cite{Leskovec:Kronecker}. Thanks to these characteristics, KPGM has been selected as a generation algorithm for the Graph 500 Supercomputer Benchmark \cite{Pinar}.

The KPGM starts with a initial square matrix $\Theta_1$ of size $b\times b$, where each cell value is a probability. To generate a graph, the algorithm uses k Kronecker multiplications to grow until a determined size (obtaining $\Theta_k$ with $b^k=N$ rows and columns). Each edge is then independently sampled using a Bernoulli distribution with parameter $\Theta_k(i,j)$. A rough implementation of this algorithm has time $O(N^2)$, but improved algorithms can generate a network in $O(M\log N)$, where M is the number of edges in the network \cite{Leskovec:Kronecker}. According to \cite{Leskovec:Kronecker}, the learning time is linear in the number of edges.

\section{Chung-Lu Model and Invariant MC Distribution}\label{invariantmc}

Define graph $G = \left<V, E\right>$, where $V$ is a set of $N$ vertices, or nodes, and $E = V\times V$ is a set of $M$ edges or relationships between the vertices.  Let $A$ represent the adjacency matrix for $G$ where:

\begin{equation}
A_{ij} =
\begin{cases}
1 & \text{if $E$ contains the tuple } (v_i, v_j) \\
0 & \text{otherwise}
\end{cases}
\end{equation}

Next, define the diagonal matrix $D$ such that:

\begin{equation}
D_{ij} =
\begin{cases}
\sum_k A_{ik} & \text{if $i = j$}\\
0 & \text{otherwise}
\end{cases}
\end{equation}

The diagonal of matrix $D$ represents the degree of each node, where $D_{ii}$ is the degree of node $i$.
Finally, define the transition probability matrix $P$:
\begin{equation}
P_{ij} =
\frac{A_{ij}}{D_{ii}}
\end{equation}

This transition probability matrix is the probability of arriving at any node during a random walk that is uniform over the edges.  It is important to note that the rows of $P$ are normalized:

\begin{equation}
\sum_j P_{ij} = \sum_j \frac{A_{ij}}{D_{ii}} = \frac{1}{D_{ii} }\sum_j A_{ij} = \frac{D_{ii}}{D_{ii}} = 1
\end{equation}

\subsection{Chung Lu Model}
The Chung Lu model assigns edges to the graph by independently laying edges for each possible edge with probability:
$$
\frac{D_{ii}D_{jj}}{2M}
$$
It is assumed $D_{kk} < \sqrt{M}\  \forall k$.  The expected degree distribution for this graph is simply:
$$
E_G[D_{ii}^{CL}] = \sum_j \frac{D_{ii}D_{jj}}{2M} = D_{ii}\sum_j\frac{D_{jj}}{2M} = D_{ii}
$$
\subsection{Fast Chung Lu Model}

The invariant distribution of a graph is the distribution that when multiplied with the transition probability matrix returns itself:
\begin{equation*}
\pi *P = \pi
\end{equation*}
A possible candidate for such a distribution is defined in terms of the degrees of the network, where $\pi(i) = \frac{D_{ii}}{2M}$: 
$$
\pi(i) = \sum_j  \pi(j) * P_{ji} = \sum_j  \frac{D_{jj}}{2M} \cdot \frac{A_{ji}}{D_{jj}} = \sum_j  \frac{A_{ji}}{2M} = \frac{D_{ii}}{2M}
$$
If we assume the matrix $P$ is stationary (not changing as the random walker steps through the graph) and non-bipartite, the $\pi$ distribution is unique and tends to the stationary distribution $\pi$ as the number of steps tends to infinity \cite{Randomwalks}.

In \cite{Pinar}, the authors describe a fast edge-laying algorithm which runs in $O(M)$.  The algorithm proceeds by creating a vector of size $O(M)$, then places the IDf of each node $v_i$ in the vector $D_{ii}$ times.  It is not hard to see that since the sum of the degrees equals the number of edges in the graph, each node can place its ID exactly $D_{ii}$ times without collision, and without leaving empty space in the vector.

Next, a node ID $v_i$ is drawn from the vector -- this can be done in $O(1)$ by drawing a uniform random variable between $1$ and $M$, and using offsets to index into the array.  The next step is to draw another \emph{independent} vertex $v_j$ from the vector and place the edge between the two sampled nodes.  In a special graph, the regular graph, we can show that the probability of an edge existing is exactly the same for the fast CL method and the slow.
\begin{algorithm}
\caption{CL($\pi,N,|E|$)}
\label{CLalg}
\begin{algorithmic}[1]
\STATE $E^{CL} = \{\}$
\STATE $initialize(queue)$
\FOR{iterations}
\IF{queue is empty}
\STATE $v_j = pi\_sample(\pi)$
\ELSE
\STATE $v_j = pop(queue)$
\ENDIF
\STATE $v_i = pi\_sample(\pi)$
\IF{$e_{ij} \not\in E^{CL}$}
\STATE $E^{CL} = E^{CL} \cup e_{ij}$
\ELSE
\STATE $push(queue, v_i)$
\STATE $push(queue, v_j)$
\ENDIF
\ENDFOR
\STATE $return(E^{TCL})$
\end{algorithmic}
\end{algorithm}

\begin{proposition}\label{edgepresent}
In a regular graph, the probability of an edge existing in the Fast Chung-Lu model is the same as the probability of an edge existing in the Slow Chung-Lu Model.
\end{proposition}
\begin{proof}
Let $v_i, v_j$ be two nodes in our network.  According to the Fast Chung Lu model we will select every node at random with replacement, meaning the number of times the node $v_j$ will be selected as the first node is $\bar{D}_j$, where $\bar{D}$ is the degree for every node in the network, and $j$ is used for notation, to indicate the particular node.  Since this graph is regular, $\bar{D}_i = \bar{D}_j\ \forall i,j$.  The probability of an edge being placed from $v_j$ to $v_i$ is the sum:

\begin{equation}
\begin{split}
P(e_{ij} | v_j) &= \frac{\bar{D}_i}{2M} + \sum_{v_k \in V, k \neq i} \frac{\bar{D}_k}{2M} \frac{\bar{D}_i}{2M - \bar{D}_k} +\dots \\
&= \frac{\bar{D}_i}{2M} + \frac{\bar{D}_i}{2M}\sum_{v_k \in V, k \neq i} \frac{\bar{D}_k}{2M - \bar{D}_k} +\dots \\
&= \frac{\bar{D}_i}{2M} + \frac{\bar{D}_i}{2M}\frac{2M - \bar{D}_i}{2M - \bar{D}_k} +\dots \\
&= 2\frac{\bar{D}_i}{2M}
\end{split}
\end{equation}

The sum continues to $\bar{D}_j$.  The probability of inserting on the $d$ insertion is therefore:
{\small
\begin{equation}
\begin{split}
&=\sum_{v_{k_1}} \frac{\bar{D}_k}{2M}  \sum_{v_{k_2}} \frac{\bar{D}_{k_2}}{2M - \bar{D}}\dots \sum_{v_{k_d}} \frac{\bar{D}_{k_d}}{2M - (d-1)\bar{D}}\frac{\bar{D}_i}{2M - d\bar{D}}\\
&=\frac{\bar{D}_i}{2M}\sum_{v_{k_1}} \frac{\bar{D}_{k_1}}{2M - \bar{D}}  \sum_{v_{k_2}} \frac{\bar{D}_{k_2}}{2M - 2\bar{D}}\dots \sum_{v_{k_d}} \frac{\bar{D}_{k_d}}{2M - d\bar{D}}\\ 
&=\frac{\bar{D}_i}{2M}
\end{split}
\end{equation}
}
where $v_{k_l} \in V, k_l \neq k_{l-1}, \dots, k_{1}, i$.  Thus each time we place an edge from $D_j$, we place it with probability $\frac{\bar{D}}{2M}$ on node $v_i$.  After $d$ insertions, the probability of having an edge $e_{ji}$ is then $d\frac{\bar{D}}{2M}$.  If we draw $M$ times for the first node, the expected number of draws on $D_j$ is then $\frac{D}{2}$, meaning the probability of connecting $v_j$ to $v_i$ is $\frac{DD}{4M}$.  If we include the opposite direction, we get $\frac{DD}{2M}$, which is the same as the probability in the slow method.
\end{proof}
Usually we do not have a regular graph, meaning the breakdown between the degrees does not have the convenient cancellation of sums like the regular graph.  However, for sparse graphs we assume the proportion of degrees is close enough to one another such that the summations effectively cancel.  The difference between the two probabilities is illustrated in Figure \ref{fig:edgeprobs}. To do this, we show the edge probabilities along the x-axis as predicted by the original CL method ($\frac{D_{ii}D_{jj}}{2M}$) versus a simulation of 10,000 networks for the fast edge probabilities.  The y-axis indicates the proportion of generated networks which have the edge (we plot the top 10 degree nodes' edges).    The dataset we use is a subset of the Purdue University Facebook network, a snapshot of the class of 2012 with approximately 2000 nodes and 15,000 edges -- using this smaller subset exaggerates the collisions and their effects on the edge probabilities.  In panel (a), we show the probabilities for the original network, where the probabilities are small and unaffected by the fast model.  

\begin{figure}
\centering
\subfloat[Original]{\includegraphics[width=.49\columnwidth]{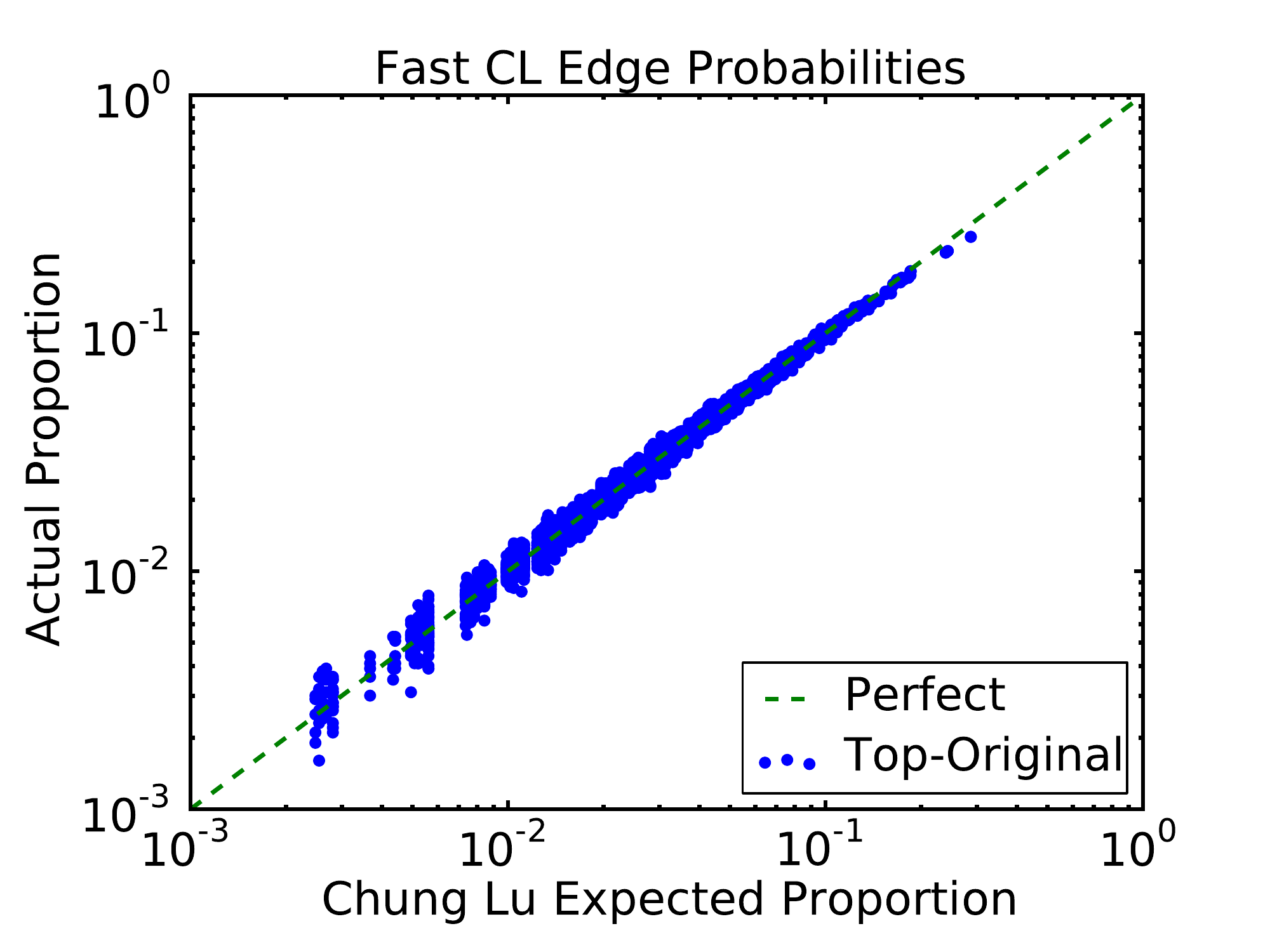}}
\subfloat[Adjusted]{\includegraphics[width=.49\columnwidth]{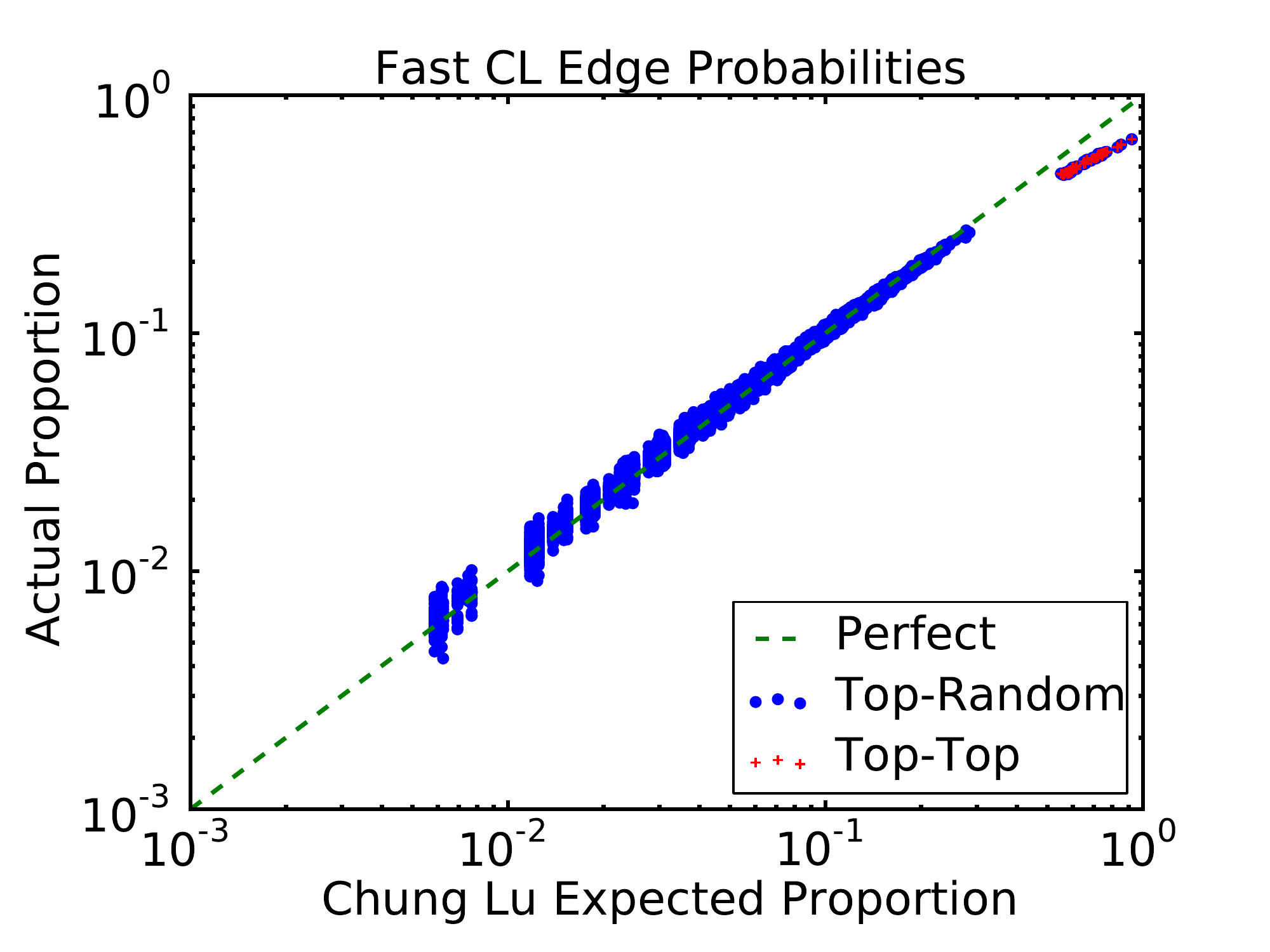}}
\caption{Fast CL edge probability (y-axis) versus Slow CL edge probability (x-axis).  In (a) we show the true Facebook 2012 Network and in (b) augmented to force high degree nodes to approach $\sqrt{2M}$ degree, and connect to each other.}
\label{fig:edgeprobs}
\end{figure}

To test the limits of the method, in panel (b) we take the high degree nodes from original network and expand them such that they have near $\sqrt{2M}$ edges elsewhere in the network.  Additionally, these high degree nodes are connected \emph{to each other}, meaning they approach the case where $\frac{D_{ii}D_{jj}}{2M} > 1$.  Another 10000 networks are generated from the fast model to match this augmented network.  We see that the randomly inserted edges still follow the predicted slow CL value, although the probabilities are slightly higher due to the increased degrees.  It is only in the far extreme case where we connect $\sqrt{2M}$ degree nodes to one another that we see a difference in the realized probability from the CL probability.  These account for $.05\%$ of edges in the augmented network, which has been created specifically to test for problem cases.  For social networks, it is unlikely for these situations to arise.

\section{Correction to Fast Model}\label{sec:fastcor}
In order to actually generate this graph it is efficient to use rejection sampling.  Namely, we draw two nodes from $\pi$ and attempt to place an edge between them.  If an edge already exists, we reject the sample and draw again.  In general, as we are using sparse graphs we will not have many \emph{collisions}, and so few samples are rejected.  

One thing to notice that the algorithm assumes we can draw an \emph{edge} only once (sampled without replacement), but \emph{nodes} are drawn multiple times (sampled with replacement).  However, the samples are rejected according to whether or not an \emph{edge} exists.  As certain nodes have a higher degree, the probability of collision is higher for them, meaning their edges are rejected more frequently than low degree nodes.  Rejection of those \emph{node samples} means that the nodes have their degree under sampled.

\begin{proposition}\label{collision}
When repeated samples of the same edge are dropped, the nodes of high degree have their degrees underestimated.
\end{proposition}
\begin{proof}
Let $v_i, v_j$ be two nodes in our network such that $\pi(i) > \pi(j)$, and let $v_k$ be a node attempting to lay an edge with another node.  Comparing the probability of collision on $v_i,v_k$ vs. $v_j,v_k$ gives us:
\begin{equation*}
\begin{split}
&P(e_{ik}) = \pi(i)\pi(k) > \pi(j)\pi(k) \\
&P(collision_{ik}) = P(e_{ik})^2 = \left(\pi(k)\pi(i)\right)^2> \left(\pi(k)\pi(j)\right)^2
\end{split}
\end{equation*}
Since more edges are laid by high degree nodes than low, this implies it is more likely for high degree nodes such as $v_i$ to experience collisions on edge insertions, biasing their expected degrees.
\end{proof}
\begin{figure}
\centering
\subfloat[Epinions]{\includegraphics[width=.49\columnwidth]{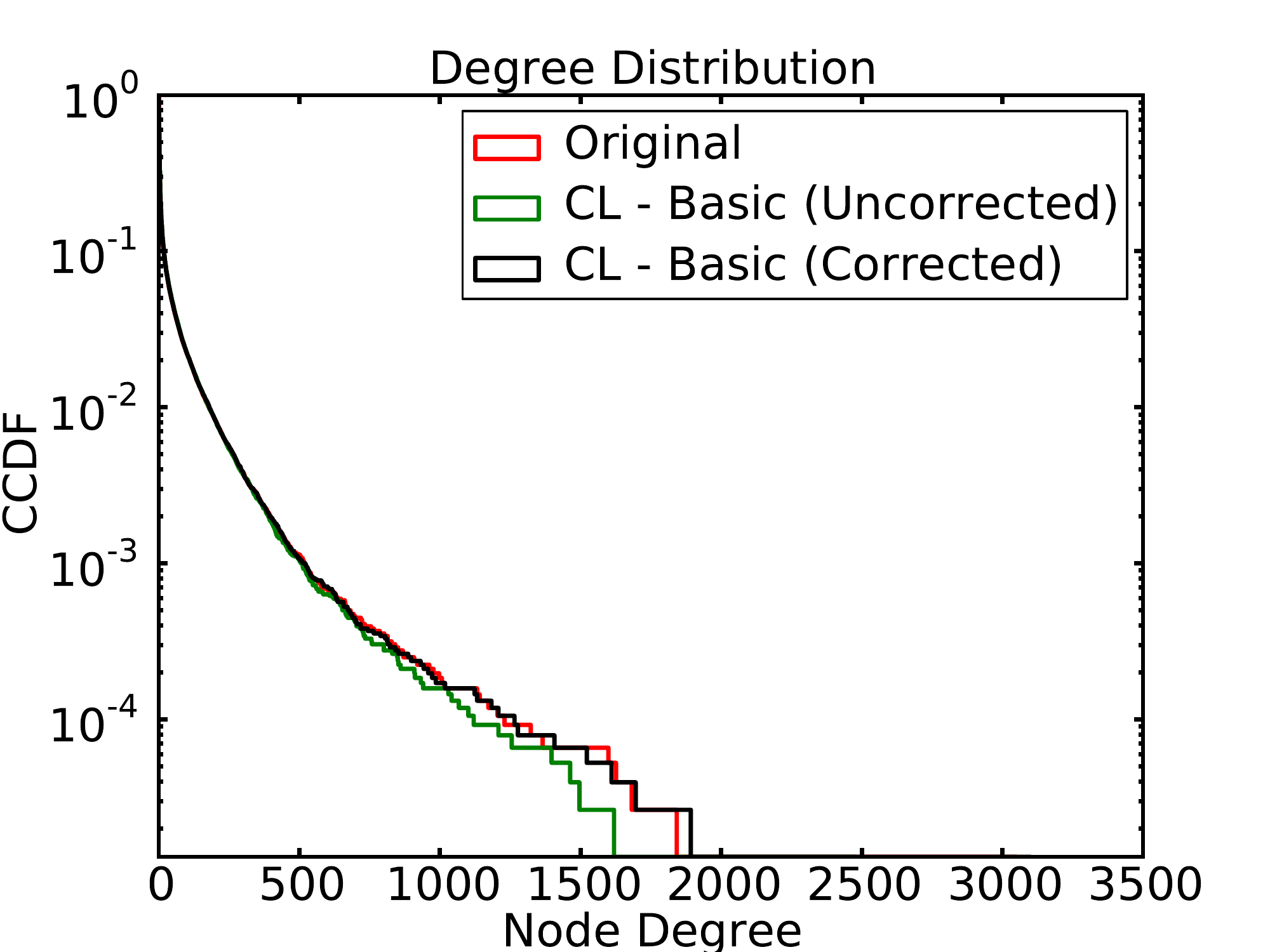}}
\subfloat[Facebook]{\includegraphics[width=.49\columnwidth]{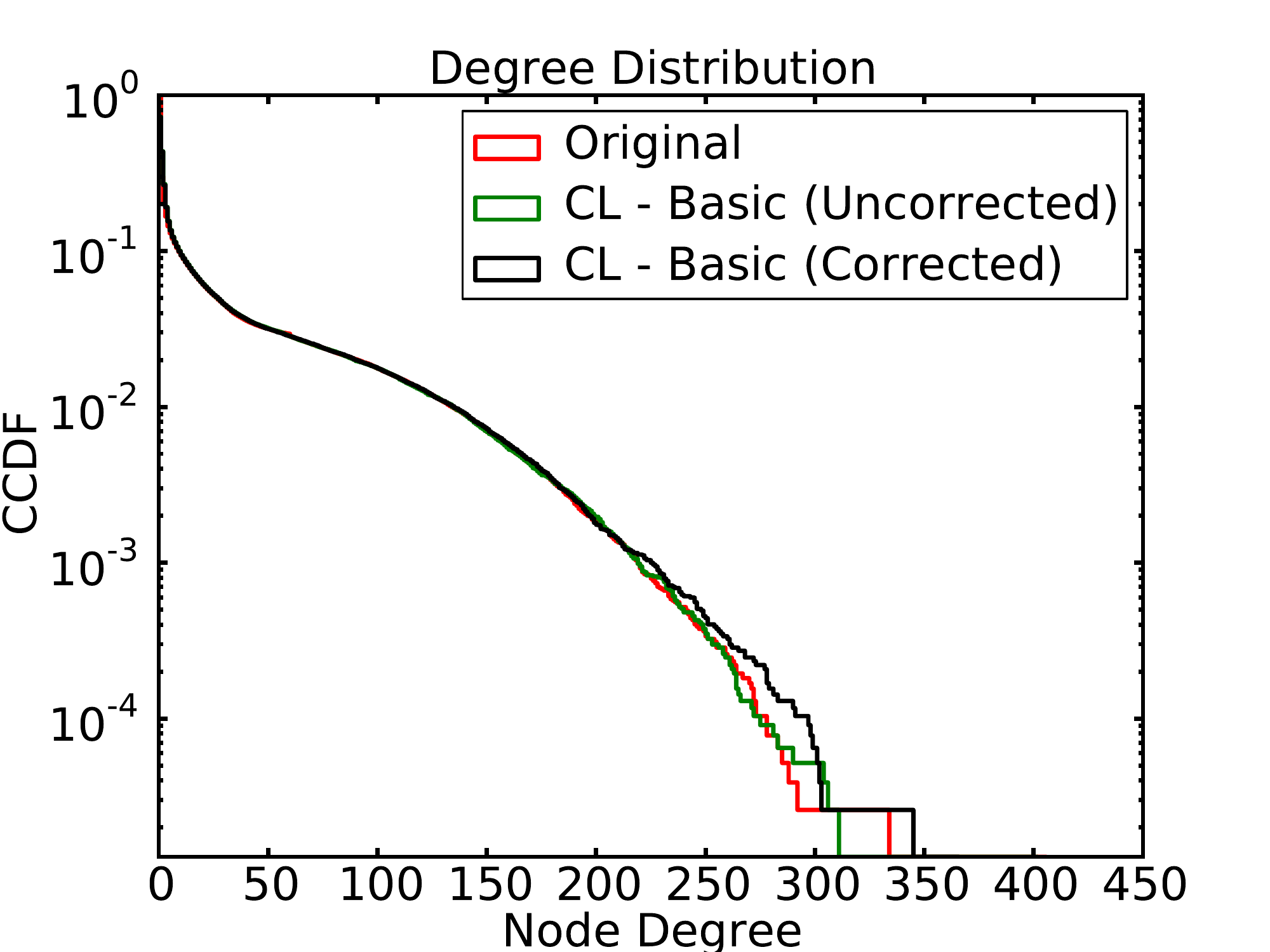}}
\caption{Comparison of Basic and Corrected CCDF on two datasets.  The original method underestimates the degree of the high degree nodes.}
\label{fig:corrected}
\end{figure}

One simple approach to correct this problem is to sample $2M$ nodes independently from $\pi$.  The can then be paired together and the pairings checked for duplicates.  Should any edges be laid more than once across a pair of nodes, the entire set of nodes is randomly permuted and rematched.  This process continues until no duplicate pairings are found.

The general idea behind this random permutation motivates our correction to the fast method.  While the random permutation of all $2M$ nodes is somewhat extreme, a method which permutes only a few edges in the graph -- the ones with collisions -- is feasible.  With this in mind, our solution to this problem is straightforward.  Should we encounter a collision, we place both vertices in a waiting queue.  Before continuing with regular insertions we will attempt to select neighbors for all nodes in the waiting queue.  Should the new edge for a node in the queue also encounter a collision, the chosen neighbor is also placed in the queue, and so forth.  This ensures that if a node is `due' for a new edge but is prevented from receiving it due to a collision, the node is `slightly permuted' by exchanging places with a node sampled later.   

 This shuffling ensures that $M$ edges actually be placed, which is needed by proposition \ref{edgepresent}, without affecting the degree distribution as can happen by proposition 2.  Furthermore, as we leave an edge if it ever occurs, the probability defined by proposition \ref{edgepresent} is never \emph{lowered} for the edges which have multiple occurrences, only raised to ensure $M$ edges are placed.

Our correction to the fast CL model assumes independence between the edge placements and the current graph configuration.  This independence only truly holds when collisions are allowed (i.e. when generating a multigraph).  In practice edge placements are not truly independent, as we disallow the placement of edges that already exist in the graph.  The correction we have described removes the bias described in proposition \ref{collision} but is not guaranteed to generate graphs exactly according to the original $\pi$ distribution.  The fast version of the graph generation algorithm must project from a space of multigraphs down into a space of simple graphs, and this projection is not necessarily uniform over the space of graphs.  However, our empirical results show that on sparse graphs our correction removes the majority of the bias due to collisions and that the bias from the projection is negligible, meaning we can treat graphs form the corrected fast generation as being drawn from the original Chung-Lu graph distribution.  While the slow Chung-Lu model is guaranteed to produce unbiased $\pi$ distributed graphs, the fast method produces graphs which are nearly indistinguishable from the slow method and runs an order of magnitude faster.

In Figure \ref{fig:corrected}, we can see the effect of the correction on two labeled datasets, Epinions and Facebook (described in section  \ref{sec:experiments}).  The green line corresponding to the simple insertion technique underestimates the degrees of the high degree nodes in both instances.  The correction results in having a much closer match on the high degree nodes.  By utilizing this correction, we can generate graphs whose degree distributions are unaffected by the possibility of collision and are able to generate graphs in $O(M)$.

\section{Transitive Chung-Lu Model}\label{sec:tcl}

\begin{algorithm}
\caption{TCL($\pi,\rho,N,|E|,iterations$)}
\label{TCLalg}
\begin{algorithmic}[1]
\STATE $E^{TCL} = CL(\pi,N,|E|)$
\STATE $initialize(queue)$
\FOR{iterations}
\IF{queue is empty}
\STATE $v_j = pi\_sample(\pi)$
\ELSE
\STATE $v_j = pop(queue)$
\ENDIF
\STATE $r = bernoulli\_sample(\rho)$
\IF{$r = 1$}
\STATE $v_k = uniform\_sample(E^{TCL}_j)$
\STATE $v_i = uniform\_sample(E^{TCL}_k)$
\ELSE
\STATE $v_i = pi\_sample(\pi)$
\ENDIF
\IF{$e_{ij} \not\in E^{TCL}$}
\STATE $E^{TCL} = E^{TCL} \cup e_{ij}$
\STATE \textcolor{Gray}{// remove oldest edge from $E^{TCL}$}
\STATE $E^{TCL} = E^{TCL} \setminus min(time(E^{TCL}))$
\ELSE
\STATE $push(queue, v_i)$
\STATE $push(queue, v_j)$
\ENDIF
\ENDFOR
\STATE $return(E^{TCL})$
\end{algorithmic}
\end{algorithm}

A large problem with the Chung-Lu model is the lack of \emph{transitivity} captured by the model.  As many social networks (among others) are formed via friendships, drawing randomly from distribution of nodes across the network fails to capture this property.  We propose the \emph{Transitive Chung Lu} model described in algorithm \ref{TCLalg}, which has a probability of a `random surfer' connecting two nodes across the network but has an additional probability of creating a new transitive edge across a pair of nodes connected by a 2 hop path.  The CL model is now a special case of TCL where $\rho=0$ and edge selection is always done through a random walk.  In the TCL model, we include the transitive edges \emph{while} maintaining the same expected invariant distribution as the CL model.  Thus the TCL model is guaranteed to have an expected degree distribution equal to that of the original network.

We begin by constructing a graph of $M$ edges using the standard Chung-Lu model as described above.  This gives us an initial edge set $E$ which has the same expected degree distribution as the original data.  We then initialize a queue which will be used to store nodes that have a higher priority for receiving an edge.  Next, we define an update step which replaces the oldest edge in the graph with a new one selected according to the TCL model and repeat this process for the specified number of iterations.  If the priority queue is not empty, we will choose the next node in the queue to be $v_j$, the first endpoint of the edge; otherwise, on line 5 we sample $v_j$ using the $\pi$ distribution.  With probability $\rho$ we will add an edge between $v_j$ and some node $v_i$ through transitive closure by choosing an intermediate node $v_k$ uniformly from $j$'s neighbors, then selecting $v_i$ uniformly from $k$'s neighbors.  In contrast, with probability $(1-\rho)$ we use the 'random surfer' method by randomly choosing $v_i$ from the graph according to the invariant distribution $\pi$.

Either method, transitive or random surfer, returns an additional node to the method to use as the other endpoint.  If the selected edge is not already part of the graph, we will add it and remove the oldest edge in the graph (continually removing the warmup CL edges).  If the selected edge is already present in the graph, we place the selected endpoint nodes into the priority queue (lines 21 and 22).  We repeat this replacement operation many times to ensure that the original graph is mostly replaced and then return the set of edges as the new graph.  In practice, we find that $M$ replacements -- enough to remove all edges generated originally by CL -- is sufficient.  

In order to show that this update operation preserves the expected degree distribution, we prove the following:
\begin{enumerate}
\item From any starting point $v_j$ the probability of a two hop walk ending on node $v_i$ is $\pi(i)$
\item The probability that TCL selects an edge $e_{ij}$ is the same as CL selecting $e_{ij}$
\item The change in the expected degree distribution after a TCL iteration is zero
\end{enumerate}
As the graph is initialized to a CL that has expected degree distribution equal to the original graph, and updates are performed that that preserve the expected degree distribution, the final graph will have the same expected degree distribution through induction. 
Our update step is a stochastic combination of two edge insertion operations: one that samples an edge using the $\pi$ distribution as in the standard CL model, and one that samples an edge based on 2 hop paths.  Naturally the CL insertion select edges based on the $\pi$ distribution by definition.  Now we will show that sampling an edge using the existing 2 hop paths also selects edges according to the $\pi$ distribution.

\begin{theorem}\label{2hop}
Starting from any node $v_j$, if the edges in the graph are distributed according to $\pi(k)\pi(i)$ and the walker traverses two hops by sampling uniformly over the edges of $v_j$ and subsequently the selected neighbor $v_k$ of $v_j$, the probability of ending this walk on node $v_i$ is $\pi(i)$.
\end{theorem}

\begin{proof}
We can represent the probability of a particular path $v_j \rightarrow v_k \rightarrow v_i$ existing in the graph as
\begin{equation}
P(path_{jki}) = \frac{D_{jj}D_{kk}}{2M}\frac{D_{kk}D_{ii}}{2M}
\end{equation}
The probability of following this path in a uniform random walk \emph{in the CL graph}, when it exists, is:
\begin{equation}
P( walk_{jki})= \frac{ 1 }{D_j^{CL} D_k^{CL}} 
\end{equation}  

To calculate the probability of a walk starting on node $j$ and ending at node $i$, we have to normalize by the probability of walking a 2 hop path from node $j$ to any other node $i'$ in the graph:
\begin{align*}
P( walk_{ji} ) &=  { \sum_{k \in V} P(path_{jki}) \cdot P( walk_{jki} ) \over \sum_{i' \in V} \sum_{k \in V} P(path_{jki'}) \cdot P( walk_{jki'} )} \\
&= \frac{ \sum_{k \in V} \frac{D_{jj}D_{kk}}{2M}\frac{D_{kk}D_{ii}}{2M} \frac{1}{D_j^{CL} D_k^{CL}}} {\sum_{i' \in V} \sum_{k \in V} \frac{D_{jj}D_{kk}}{2M}\frac{D_{kk}D_{i'i'}}{2M} \frac{1}{D_j^{CL} D_k^{CL}} }\\
&= \frac{ \sum_{k \in V}  D_j D_k^2 D_i \frac{1}{D_j^{CL} D_k^{CL}} } {\sum_{i' \in V} \sum_{k \in V}D_j D_k^2 D_{i'}  \frac{1}{D_j^{CL} D_k^{CL}} } \\
&= \frac{ \sum_{k \in V} D_k D_i \frac{1}{D_k^{CL}}} {\sum_{i' \in V} \sum_{k \in V} D_k D_{i'} \frac{1}{D_k^{CL}} } \\
&= \frac{ D_i \sum_{k \in V} D_k \frac{1}{D_k^{CL}}} {\sum_{i' \in V} D_{i'} \sum_{k \in V} D_k \frac{1}{D_k^{CL}} } \\
&= \frac{D_i}{\sum_{i' \in V} D_i'}= \frac{D_i}{2M} = \pi(i)
\end{align*}







So regardless of the starting node $j$, the probability of landing on $i$ after traveling 2 hops uniformly over the edges is $\pi(i)$.
\end{proof}

Utilizing the above theorem, we next show the probability of an edge $e_{ij}$ existing in the graph is $\pi(j)\pi(i)$.

\begin{theorem}
The Transitive Chung-Lu model selects edge $e_{ij}$ for insertion with probability
\begin{equation}
\begin{split}
P(e_{ij}) = \pi(i) * \pi(j)
\end{split}
\end{equation}

\end{theorem}

\begin{proof}


The inductive step randomly selects a node $v_j$ from the invariant distribution $\pi$ to be the first endpoint of a new edge.  From $v_j$, we have two options to complete the edge: with probability $\rho$ we use the transitive closure to walk 2 hops to find the other endpoint, and with probability $1-\rho$ we perform a random surf using $\pi$.  The invariant distribution for $\pi^{TCL}(i)$ can then be written as:


$$
\pi^{TCL}(i) = \sum_j \pi(j) \left[ \rho * P(walk_{ji})  + (1-\rho) \pi(i) \right] \\
$$

In theorem \ref{2hop} we showed that $P(walk_{ji}) = \pi(i)$.  Now the probability of selecting edge $e_{ij}$ can be written as:

\begin{equation*}
\begin{split}
P(e_{ij}) &= \pi(j) * \pi^{TCL} \\
&= \pi(j) * (\rho * \pi(i) + (1 - \rho) * \pi(i))\\
&= \pi(j) * \pi(i)
\end{split}
\end{equation*}
\end{proof}

Therefore, the inductive step of TCL will place the endpoints of the new edge according to $\pi$. 

\begin{corollary}
The expected degree distribution of the graph produced by TCL is the same as the degree distribution of the input graph.
\end{corollary}
\begin{proof}
The inductive step of TCL places an edge with endpoints distributed according to $\pi$, so the expected increase in the degree of any node $v_i$ is $\pi(i)$.  However, the inductive step will also remove the oldest edge that was placed into the network.  Since the oldest edge can only have been placed in the graph through a Chung-Lu process or a transitive closure, the expected decrease in the degree is also $\pi(i)$, which means the expected change in the degree distribution is zero.  Because the CL initialization step produces a graph with expected degree distribution equal to the input graph's distribution, and the TCL update step causes zero expected change in the degree distribution the output graph of the TCL algorithm has expected degree distribution equal to the input graph's distribution by induction.  
\end{proof}
This means we are placing edges according to $\pi(i)\pi(j)$, and doing $M$ insertions.  This is the same model as shown for the fast CL method, which also inserts $M$ edges according to $\pi(i)\pi(j)$, meaning that if the fast CL method follows slow CL, TCL does as well.  In practice, TCL and CL capture the degree distribution well (section \ref{sec:experiments}).

\subsection{Fitting Transitive Chung Lu}
Now that we have introduced a $\rho$ parameter which controls the proportion of transitive edges in the network we need a method for learning this parameter from the original network.  For this, we need to estimate the probability $\rho$ by which edge formation is done by triadic closure, and the probability $1 - \rho$ by which the random surfer forms edges.  We can accomplish this estimation using an Expectation Maximization algorithm.  First, let $z_{ij} \in Z$ be latent variables on each $e_{ij} \in E$ with values $z_{ij} \in \{ 1, 0\}$, where $1$ indicates the edge $e_{ij}$ was laid by a transitive closure and $0$ indicates the edge was laid by a random surfer.  Although the $Z$ values are unknown we can jointly estimate them with $\rho$ using EM.

We can now define the conditional probability of placing an edge $e_{ij}$ from starting node $v_j$ given the method $z_{ij}$ by which the edge was placed:

\begin{equation*}
\begin{split}
P\left(e_{ij} | z_{ij} = 1,  v_j, \rho^t\right) =& \rho^t\sum_{v_k \in e_{j*}} \frac{\mathbb{I}[v_i \in e_{k*}]}{D_{jj}} \frac{1}{D_{kk}} \\
P\left(e_{ij} | z_{ij} = 0,  v_j, \rho^t\right) =& (1-\rho^t)\cdot\pi(i) \\
\end{split}
\end{equation*}

Given the starting node $v_j$, the probability of the edge existing between $v_i$ and $v_j$, given that the edge was placed due to a triangle closure is $\rho$ times the probability of walking from $v_j$ to a mutual neighbor of $i$ and $j$ and then continuing the walk on to $i$, while $1-\rho$ is the probability the edge was placed by a random surfer.  We now show the EM algorithm.

\subsubsection*{Expectation}
Note that the conditional probability of $z_{ij}$, given the edge $e_{ij}$ and $\rho$, can be defined in terms of the probability of an edge being selected by the triangle closure divided by the probability of the edge being laid by any method.  Using Bayes' Rule, our conditional distribution on $Z$ is simply:
\begin{equation*}
\begin{split}
&P\left(z_{ij} = 1| e_{ij},  v_j, \rho^t\right) \\
&=\frac{ \rho^t\left[\sum_{v_k \in e_{j*}} \frac{\mathbb{I}[v_i \in e_{k*}]}{D_{jj}} \frac{1}{D_{kk}}\right]}{ \rho^t\left[\sum_{v_k \in e_{j*}} \frac{\mathbb{I}[v_i \in e_{k*}]}{D_{jj}} \frac{1}{D_{kk}}\right] +  (1-\rho^t)\left[\pi(i)\right]} \\
\end{split}
\end{equation*}

And our expectation of $z_{ij}$ is
\begin{equation*}
\mathbf{E}[z_{ij} | \rho^t] = P\left(z_{ij} = 1| e_{ij},  v_j, \rho^t\right)
\end{equation*}

\subsubsection*{Maximization}
To maximize this expectation, we note that $\rho$ is a Bernoulli variable representing $P(z_{ij}=1)$.  We sample a set of edges $\mathbb{S}$ uniformly from the graph to use as evidence when updating $\rho$.  The variables $z_{ij}$ are \emph{conditionally independent} given the edges and nodes in the graph, meaning the MLE update to $\rho$ is then calculating the expectation of $z_{ij} \in \mathbb{S}$ and then normalizing over the number of edges in $\mathbb{S}$:

\begin{equation*}
\rho^{t+1}= \frac{1}{\left| \mathbb{S} \right|}\sum_{z_{ij} \in \mathbb{S}} \mathbf{E}[z_{ij} | \rho^t] 
\end{equation*}

The method we used to sample these edge subsets was to select them uniformly from the set of all edges.  This can be done quickly using the node ID vector we constructed for sampling from the $\pi$ distribution.  As any node $i$ appears $D_i$ times in this vector, sampling a node from the vector and then uniformly sampling one of its edges gives us a $\frac{D_i}{M} * \frac{1}{D_i} = \frac{1}{M}$ probability of sampling any given edge.  We gathered subsets of 10000 edges per iteration and our EM algorithm converges in just a few seconds, even on datasets with millions of edges.  Figure \ref{fig:convergences} shows the convergence time on each of the datasets.

\begin{figure}
\centering
\subfloat[]{\includegraphics[width=.49\columnwidth]{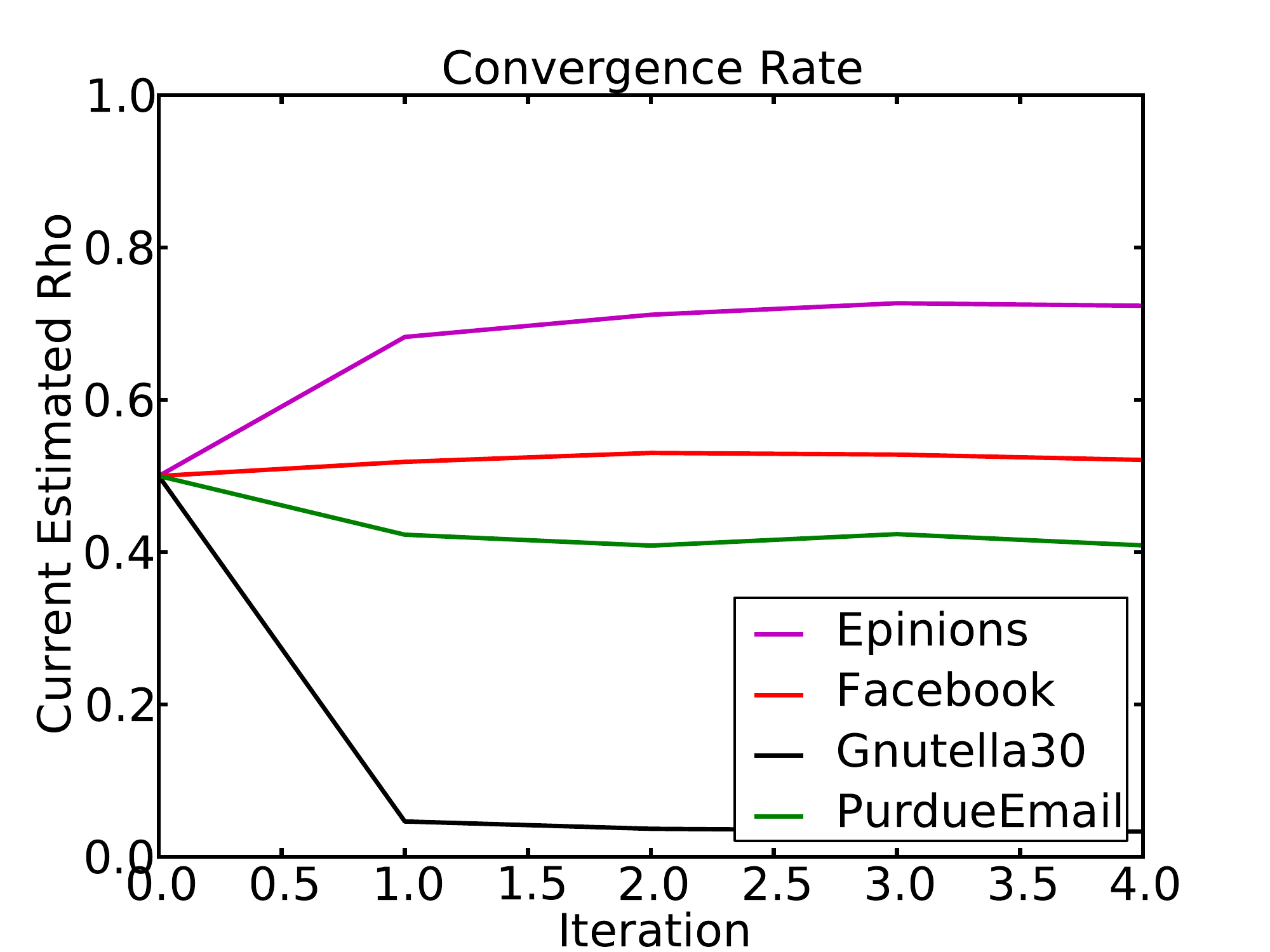}}
\subfloat[]{\includegraphics[width=.49\columnwidth]{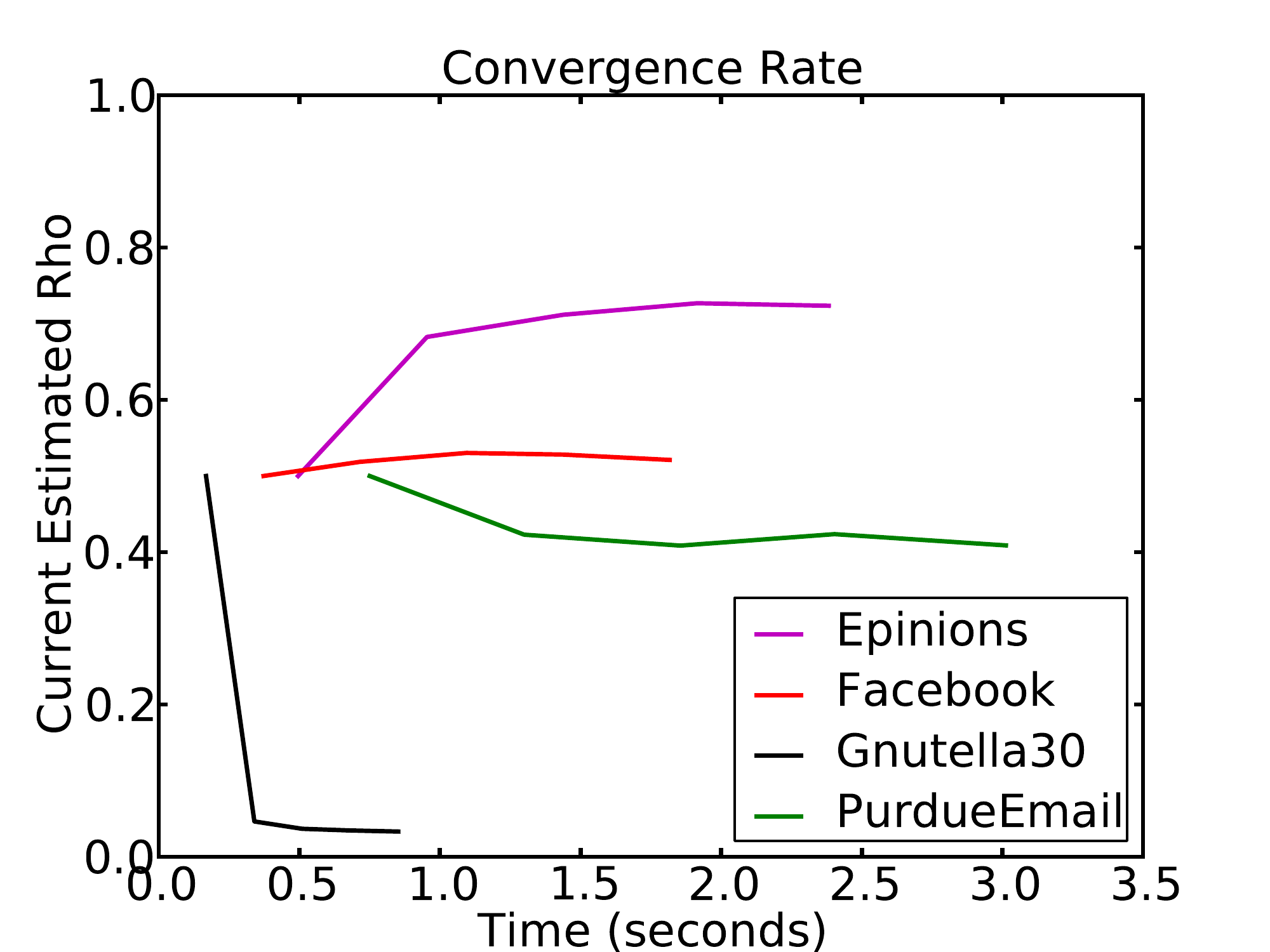}}
\caption{Convergences of the EM algorithm -- both in terms of time and number of iterations.  10000 samples per iteration.}
\label{fig:convergences}
\end{figure}

\section{Time Complexity}\label{sec:tc}
The methods presented for both generating a new network and for learning the parameter $\rho$ can be done in an efficient manner.

First, we need to bound the expected number of attempts to insert an edge into the graph.  Note that a node $v_i$ with $c$ edges has probability $\pi(i) = \frac{c}{M}$ of hitting its own edge on the draw from $\pi$; by extension, the probability of hitting its own edges $k$ times is $\pi(i)^k$.  This represents a geometric distribution which has the expected value of hits $H$ on the edges of the nodes being:
$$
\mathbf{E}\left[H|\pi(i)\right] = (1-\pi(i))\sum_{k=1}\pi(i)^{k-1} = \frac{1}{1-\pi(i)}
$$
This shows the expected number of attempts to insert an edge is bounded by a constant. 
As a result, we can generate the graph in $O(N+M)$, the same complexity as Chung Lu.  The initial steps of initializing our vector of node ids and running the basic CL model takes $O(N+M)$.  Next, we need to generate $M$ insertions while gradually removing the current edges.  This can be seen in lines 3-24 of Algorithm \ref{TCLalg}.  In this loop, the longest operations are selecting randomly from neighbors or removing an edge.  Both of these operations cost is in terms of the maximum degree of the network, which we assumed bounded, meaning those operations can be done in $O(1)$ time.  As a result, the total runtime of graph generation is $O(N+M)$.

For the learning algorithm, assume we have $I$ iterations which gather $s$ samples.  It is $O(1)$ to draw a node from the graph and $O(1)$ to choose a neighbor, meaning each iteration costs $O(s)$.  Coupled with the cost of creating the initial $\pi$ sampling vector, the total runtime is then $O(N+M+I\cdot s)$.

\begin{figure*}
\centering
\subfloat[Size]{
\small{
\begin{tabular}{|c|c|c|l|} \hline
Dataset & Nodes & Edges \\ \hline
Epinions & 75,888 & 811,480 \\ \hline
Facebook & 77,110& 500,178\\ \hline
Gnutella30 & 36,682& 176,656\\ \hline
PurdueEmail & 214,773 & 1,711,174\\ \hline
\end{tabular}}}\hfill
\subfloat[Learning Time]{
\small{
\begin{tabular}{|c|c|c|l|} \hline
Dataset & CL & KPGM & TCL \\ \hline
Epinions & N/A & 9,105.4s &  2.5s\\ \hline
Facebook & N/A& 5,689.4s & 2.0s \\ \hline
Gnutella30 & N/A& 3,268.4s & 0.9s\\ \hline
PurdueEmail &N/A & 8,360.7s & 3.0s\\ \hline
\end{tabular}}}\hfill
\subfloat[Generation Time]{
\small{
\begin{tabular}{|c|c|c|l|} \hline
Dataset & CL & KPGM & TCL \\ \hline
Epinions & 20.0s & 151.3s & 64.6s \\ \hline
Facebook &14.2s & 92.4s  & 30.8s \\ \hline
Gnutella30 & 4.2s & 67.8s & 7.0s\\ \hline
PurdueEmail & 61.0s & 285.6s & 141.0s \\ \hline
\end{tabular}}}
\caption{Dataset sizes, along with learning times and running times for each algorithm}
\label{fig:datasettable}
\end{figure*}

\section{Experiments}\label{sec:experiments}
For our experiments, we compared three different graph generating models.  The first is the fast Chung Lu (CL) generation algorithm with our correction for the degree distribution.  The second is Kronecker Product Graph Model (KPGM) implemented with code taken from the SNAP library\footnote{SNAP: Stanford Network Analysis Project. Available at http://snap.stanford.edu/snap/index.html} calculated by the authors\cite{Leskovec:Kronecker}.  Lastly, we compared the Transitive Chung Lu (TCL) method presented in this paper using the EM technique to estimate the $\rho$ parameter.  All experiments were performed in Python on a Macbook Pro, aside from the KPGM parameters which were generated on a desktop computer using C++\footnote{SNAP is written in C++}.  All of these networks were made undirected by reflecting the edges in the network, except for the Facebook network which is already undirected.

\subsection{Datasets}
To empirically evaluate the models, we learned model parameters from real-world graphs and then generated new graphs using those parameters.  We then compared the network statistics of the generated graphs with those of the original networks.  The four networks used are all large social networks, and their node and edge counts can be found in Figure \ref{fig:datasettable}.a.

The first dataset we analyze is the Epinions dataset \cite{EpinionsDataset}.  This network represents the users of Epinions, a website which encourages users to indicate other users whose consumer product reviews they `trust'.  The reviews of all users on a product are then weighted to incorporate both the reviewer ratings and the amount of trust received from other users.  The edge set of this network represents nominations of trustworthy individuals between the users.

Next, we study the collection of Facebook friendships from the Purdue University Facebook network.  In this network, the users can add each other to their lists of friends and so the edge set represents a friendship network.  This network has been collected over a series of snapshots for the past 4 years; we use nodes and friendships aggregated across all snapshots.

The Gnutella30 network is a different type than the other networks presented.  Gnutella is a Peer2Peer network where users are attempting to find seeds for file sharing \cite{Gnutella30}.  The user reaches out to its current peers, querying if they have a file.  If not, the friend refers them to other users who might have a file, repeating this process until a seed user can be found.  Because this network represents the structure of a file sharing program rather than true social interactions, it has significantly less clustering than the other networks.

Lastly, we study a collection of emails gathered from the SMTP logs of Purdue University \cite{Ahmed}.  This dataset has an edge between users who sent e-mail to each other.  The mailing network has a small set of nodes which sent out mail at a vastly greater rate than normal nodes; these nodes were most likely mailing lists or automatic mailing systems.  In order to correct for these `spammer' nodes, we remove nodes with a degree greater than $1,000$ as these nodes did not represent participants in any kind of social interaction.  The network has over two hundred thousand nodes, and nearly two million edges (Figure \ref{fig:datasettable}.a).

\subsection{Running Time}
In Figure \ref{fig:convergences} we can see the convergence of the EM algorithm when learning parameter $\rho$, both in terms of the number of iterations and in terms of the total clock runtime.  Due to the independent sample sets used for each iteration of the algorithm, we can estimate whether the sample set in each iteration is sufficiently large.  If the sample size is too small the algorithm will be susceptible to variance in the samples and will not converge.  Using Figure \ref{fig:convergences}.a we can see that after 5 iterations of 10,000 samples each our EM method has converged to a smooth line.

In addition to the convergence in terms of iterations, in Figure \ref{fig:convergences}.b we plot the wall time against the current estimated $\rho$.  The gap between $0$ and the start of the colored lines indicates the amount of overhead needed to generate our degree distribution statistic and $\pi$ sampling vector for the given graph (a step also needed by CL).  The Purdue Email network has the longest learning time at $3$ seconds.  For the same Email network, learning the KPGM parameters took approximately 2 hours and 15 minutes, so our TCL model can learn parameters from a network significantly faster than the KPGM model.

Next, the performance in terms of graph generation speed is tested, shown in Figure \ref{fig:datasettable}.c.  The maximum time taken to generate a graph by CL is $61$ seconds for the Purdue Email dataset, compared to 141 seconds to generate via TCL.  Since TCL must initialize the graph using CL and then lay its own edges, it is logical that TCL requires at least twice as long as CL.  The runtimes indicate that the transitive closures cost little more in terms of generation time compared to the CL edge insertions.  KPGM took 285 seconds to generate the same network.  The discrepancy between KPGM and TCL is the result of the theoretical bounds of each -- KPGM takes $O(M\log N)$ while TCL takes $O(M)$.

\begin{figure*}
\centering
\subfloat[Epinions]{
\includegraphics[width=.33\textwidth]{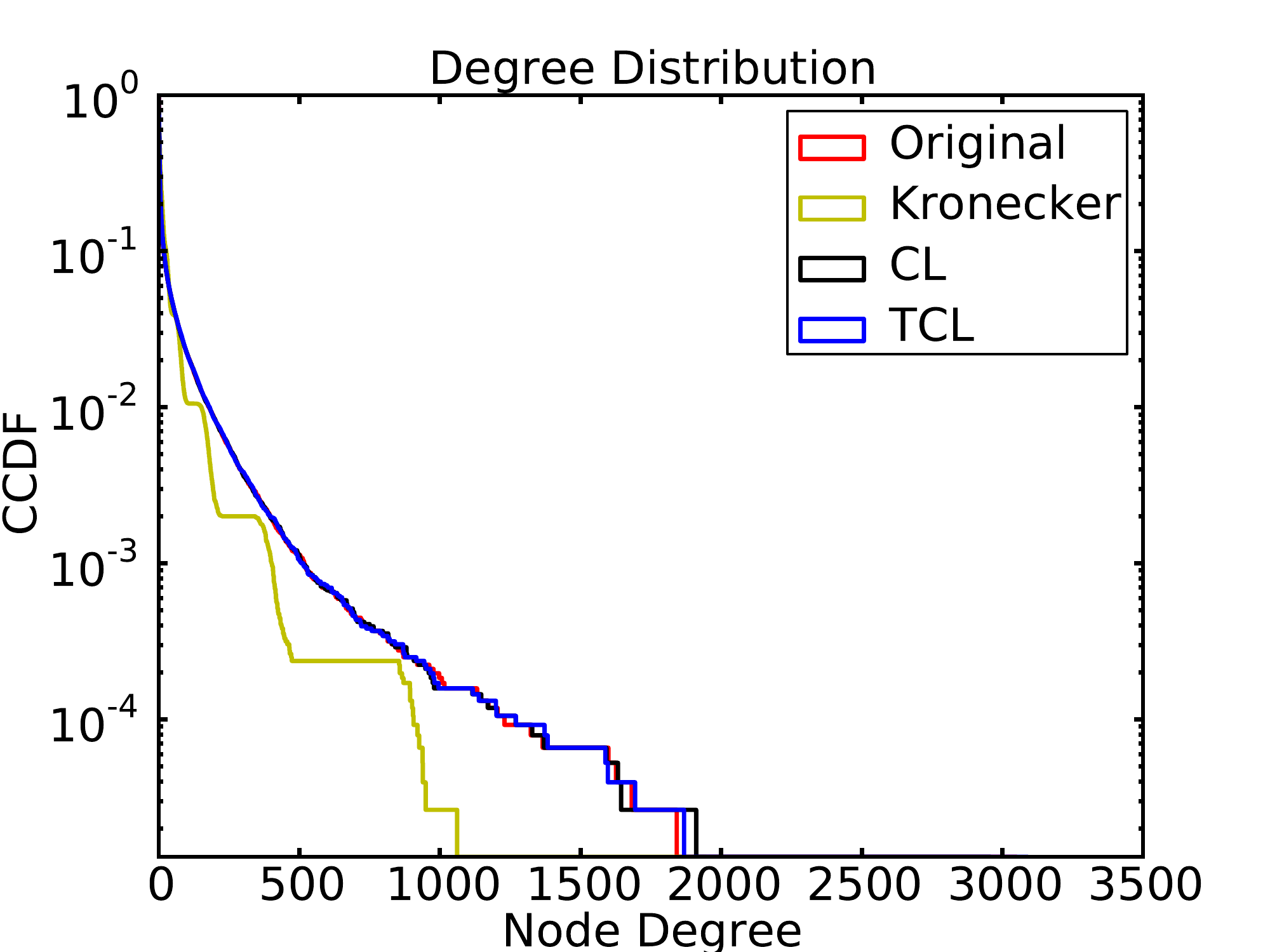}
\includegraphics[width=.33\textwidth]{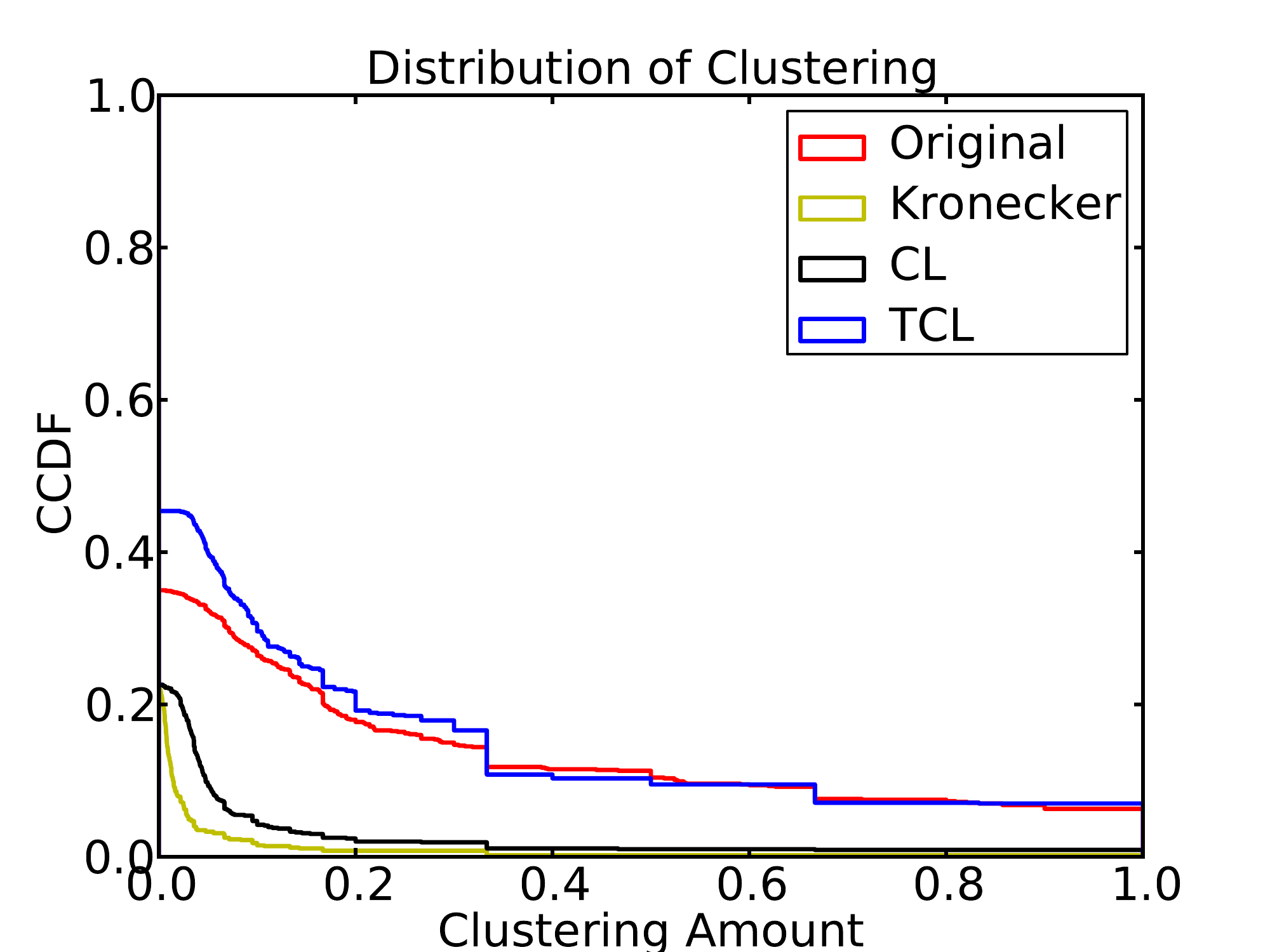}
\includegraphics[width=.33\textwidth]{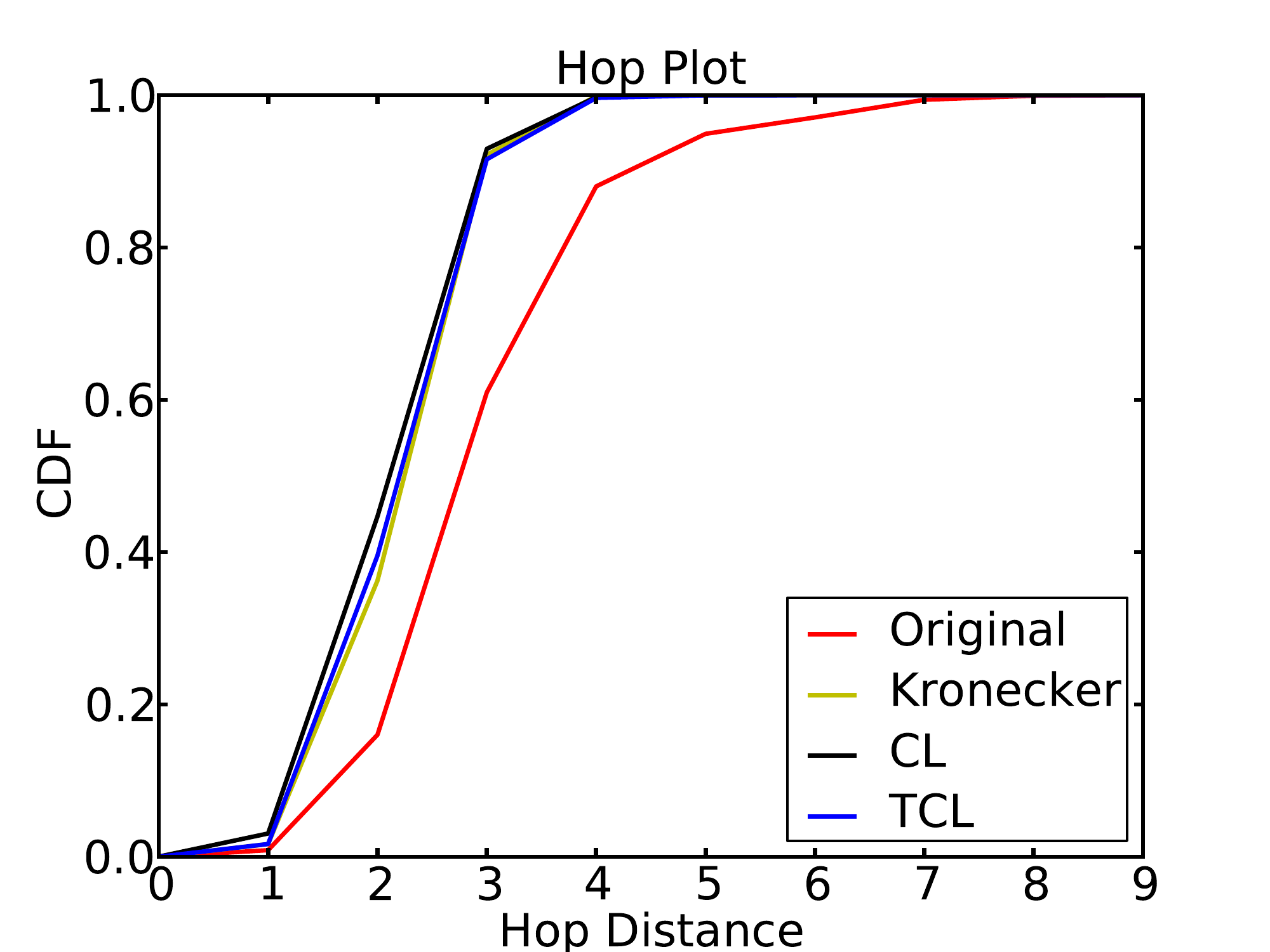}}

\subfloat[Facebook]{
\includegraphics[width=.33\textwidth]{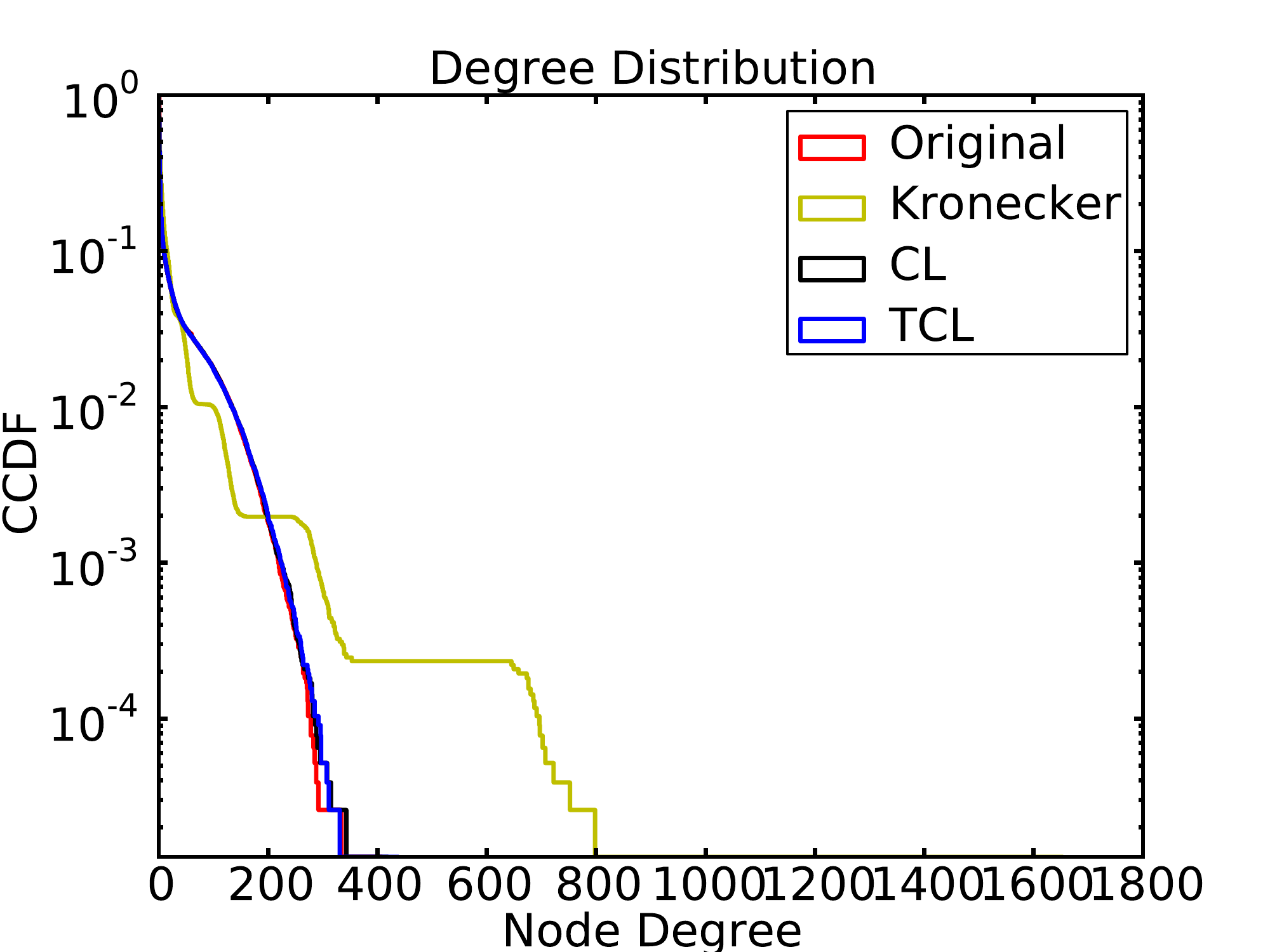}
\includegraphics[width=.33\textwidth]{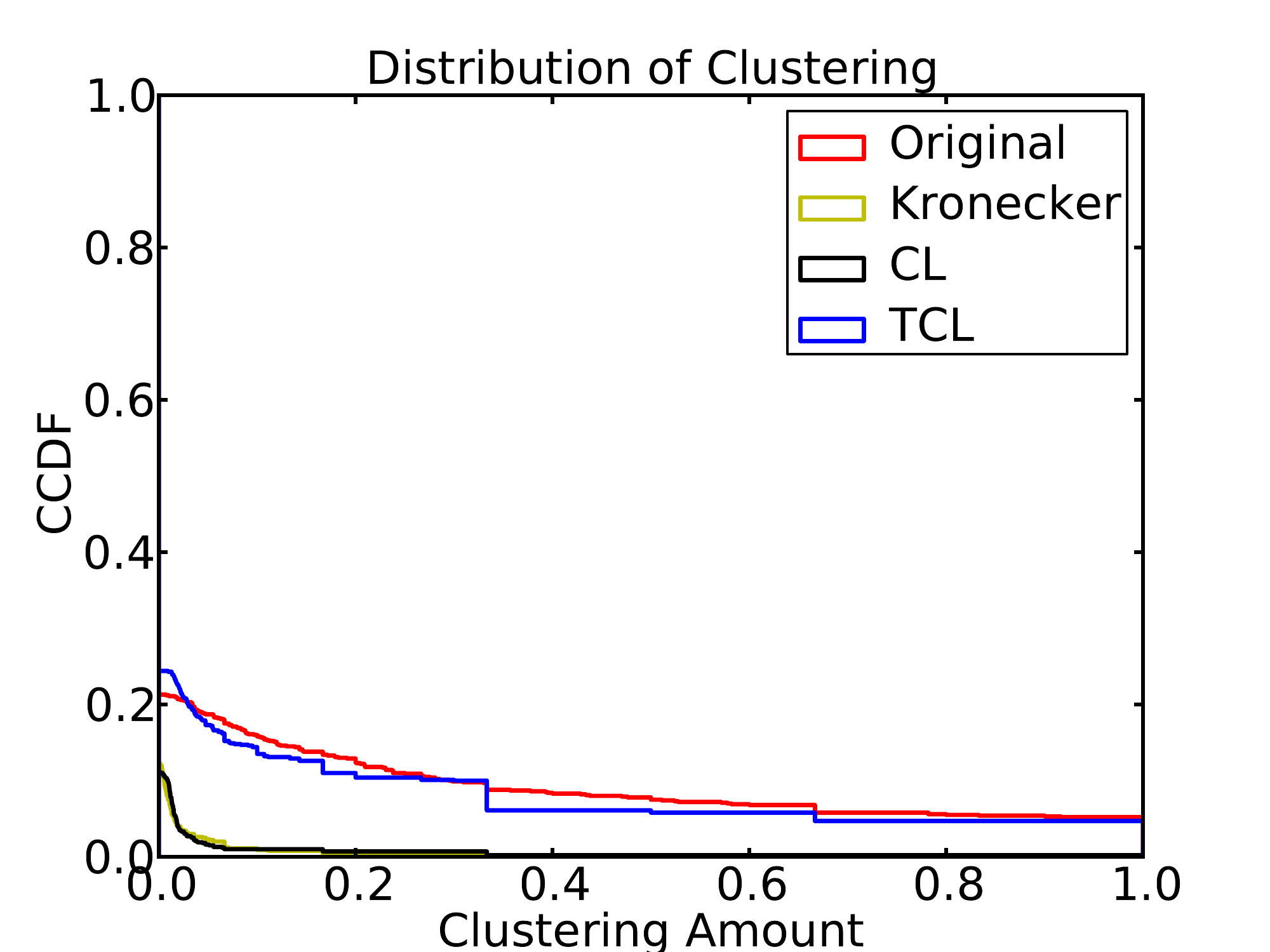}
\includegraphics[width=.33\textwidth]{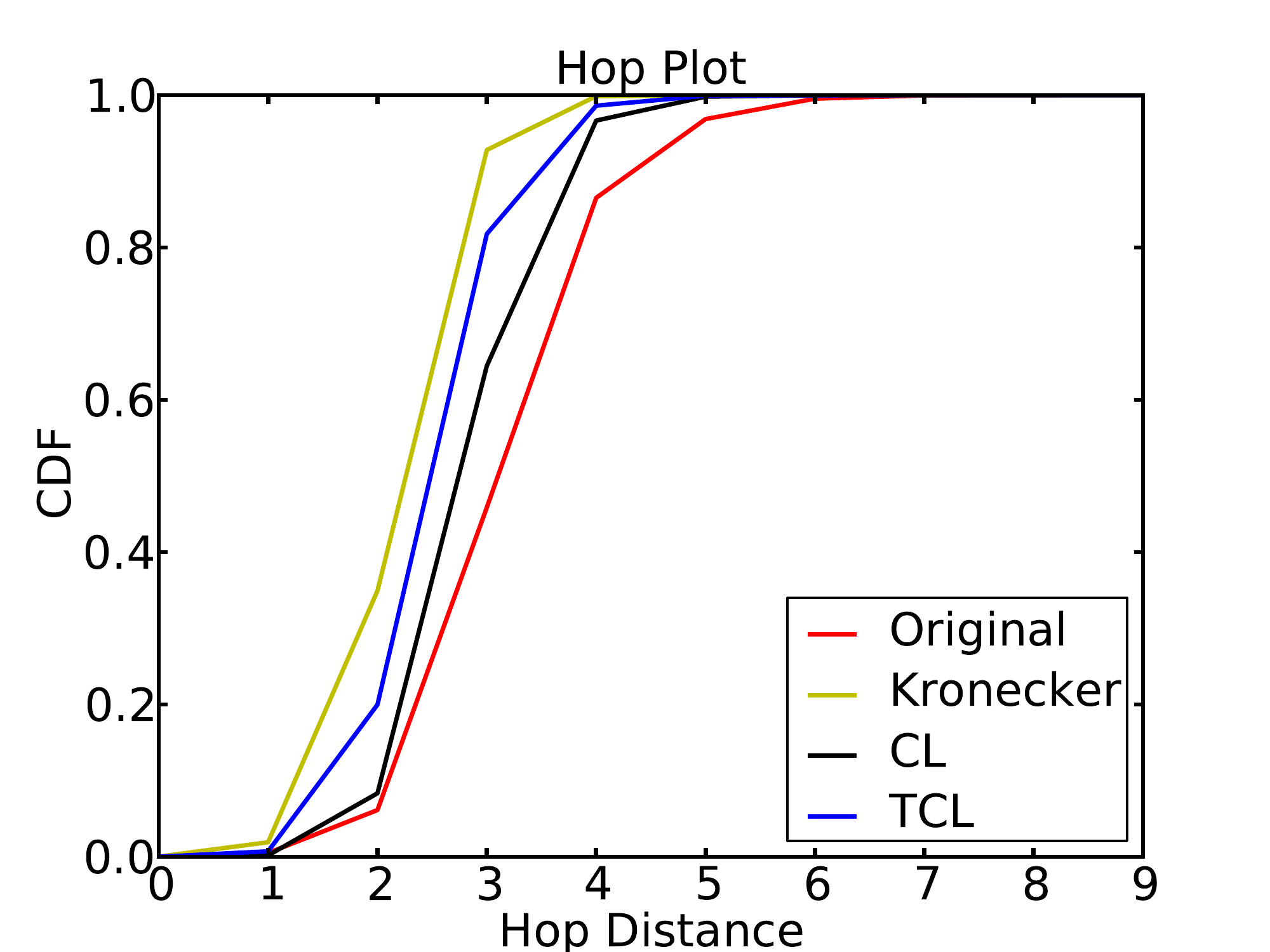}}

\subfloat[Gnutella30]{
\includegraphics[width=.33\textwidth]{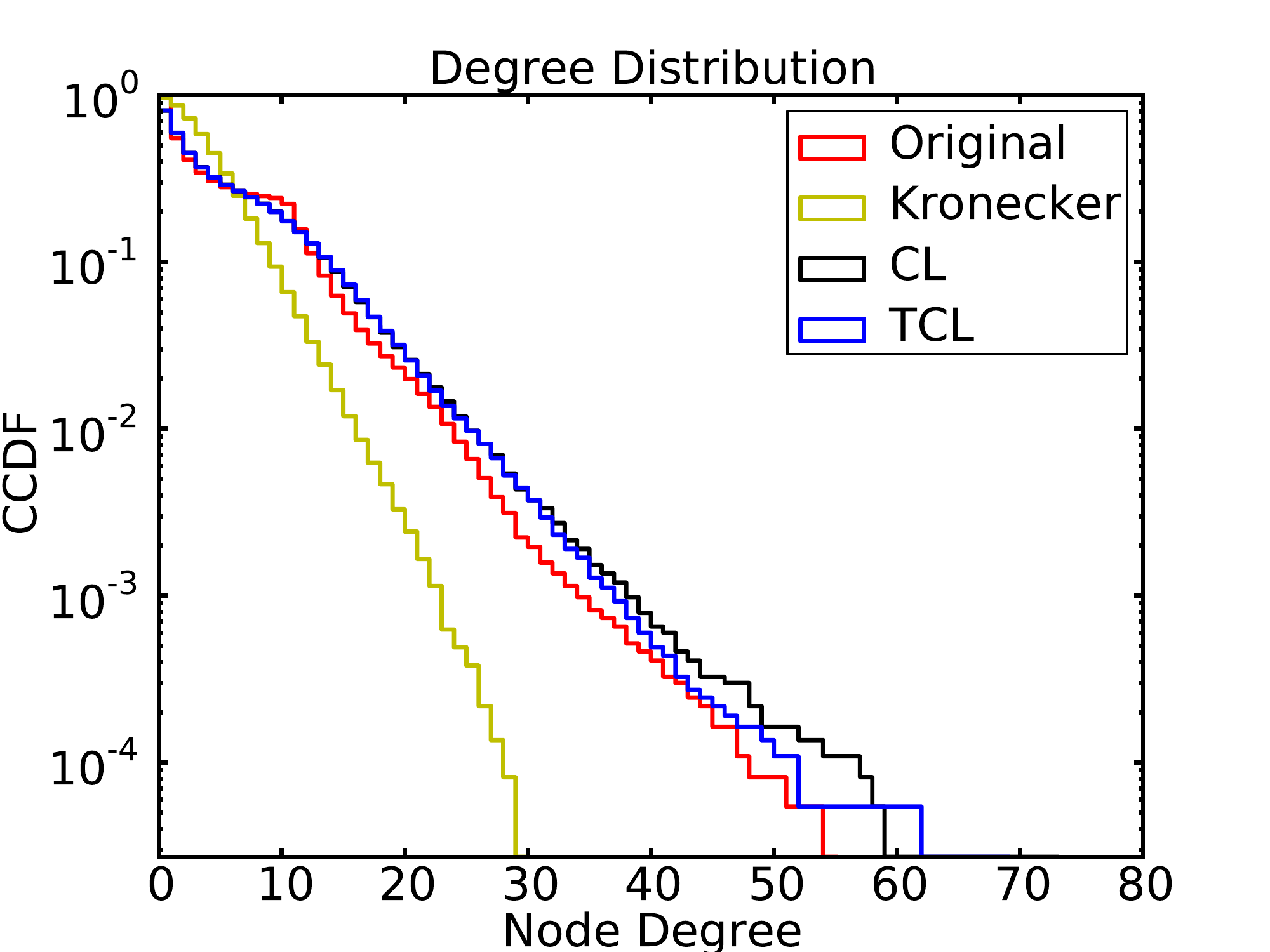}
\includegraphics[width=.33\textwidth]{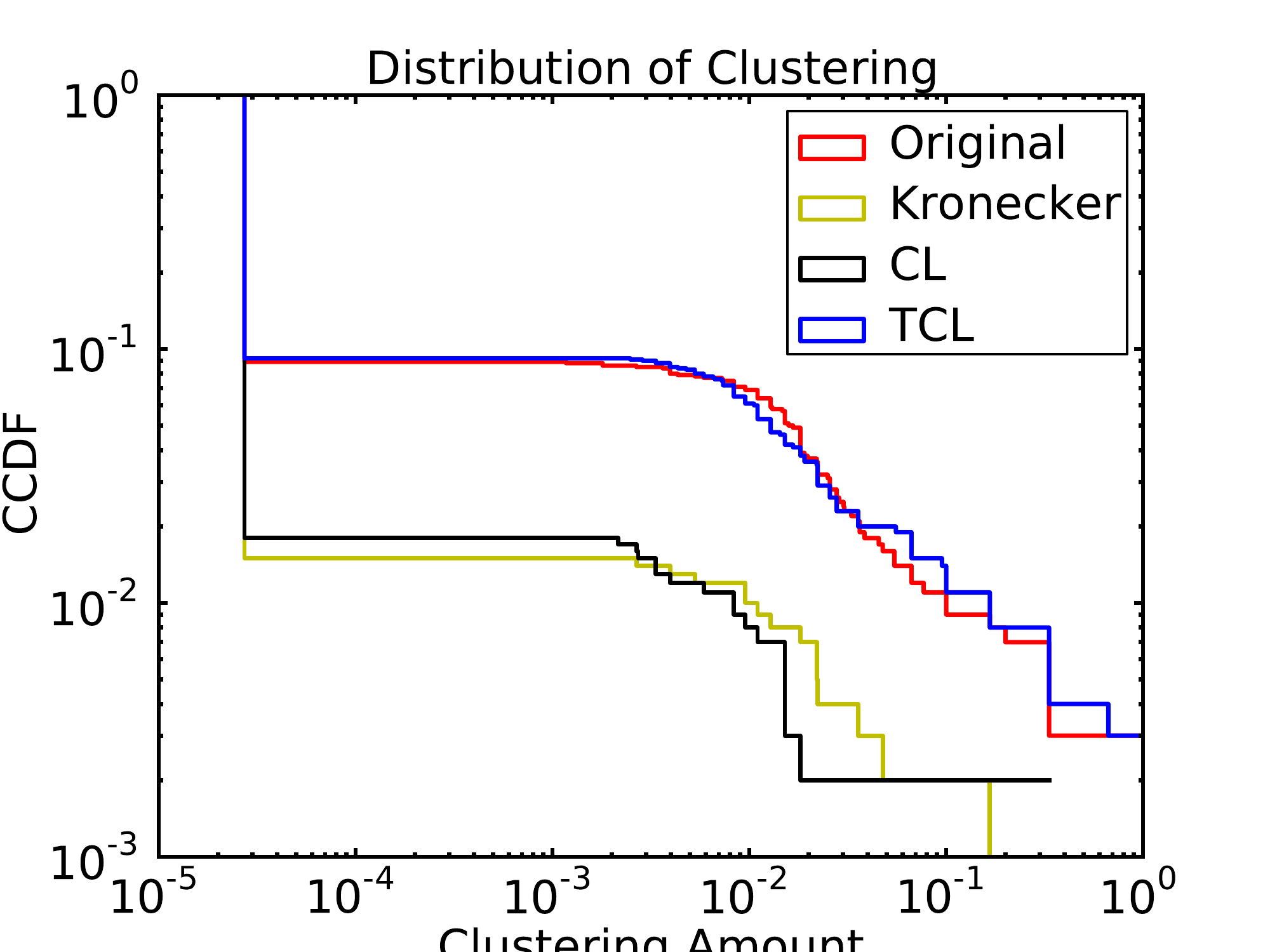}
\includegraphics[width=.33\textwidth]{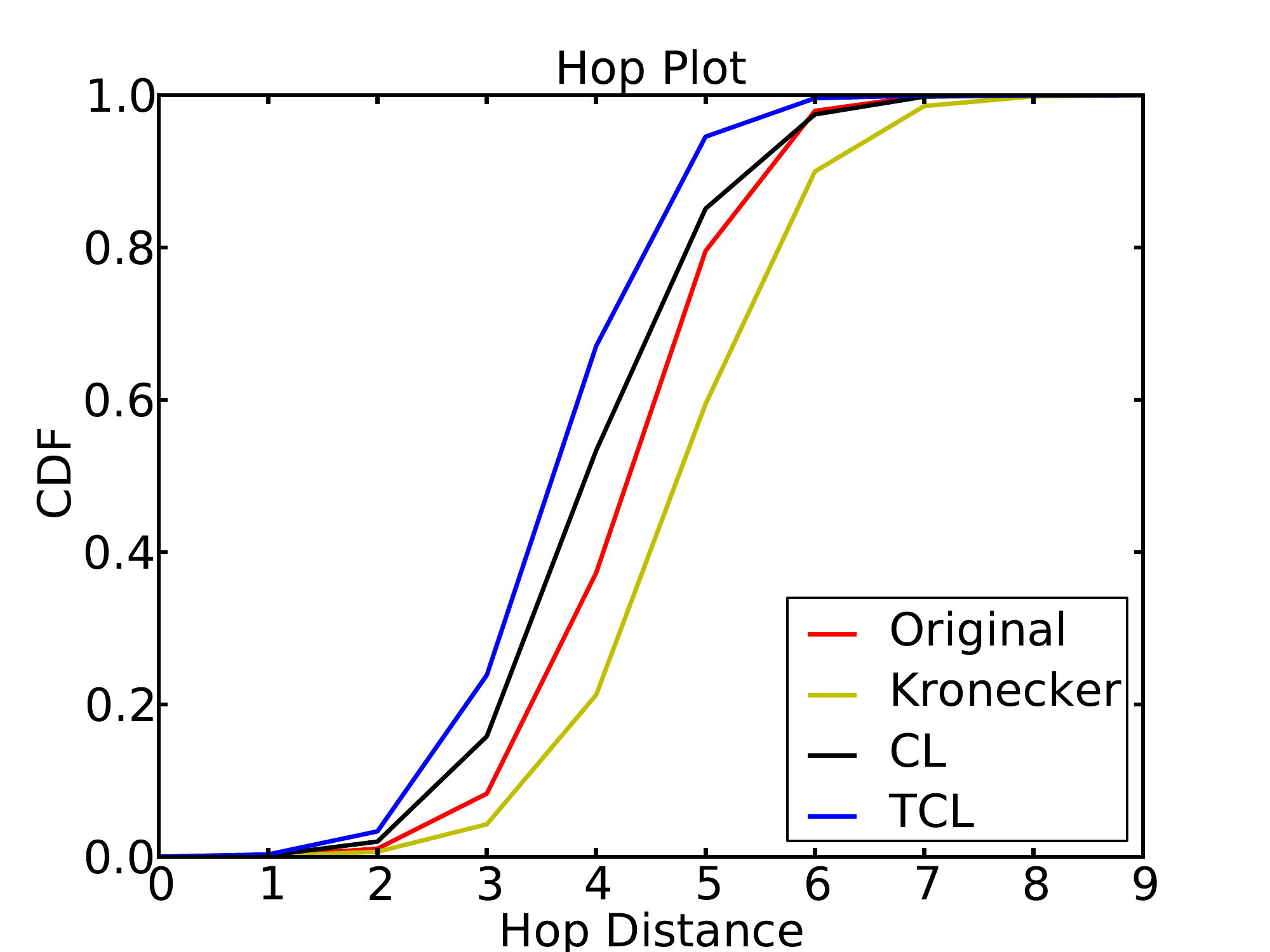}}

\subfloat[PurdueEmail]{
\includegraphics[width=.33\textwidth]{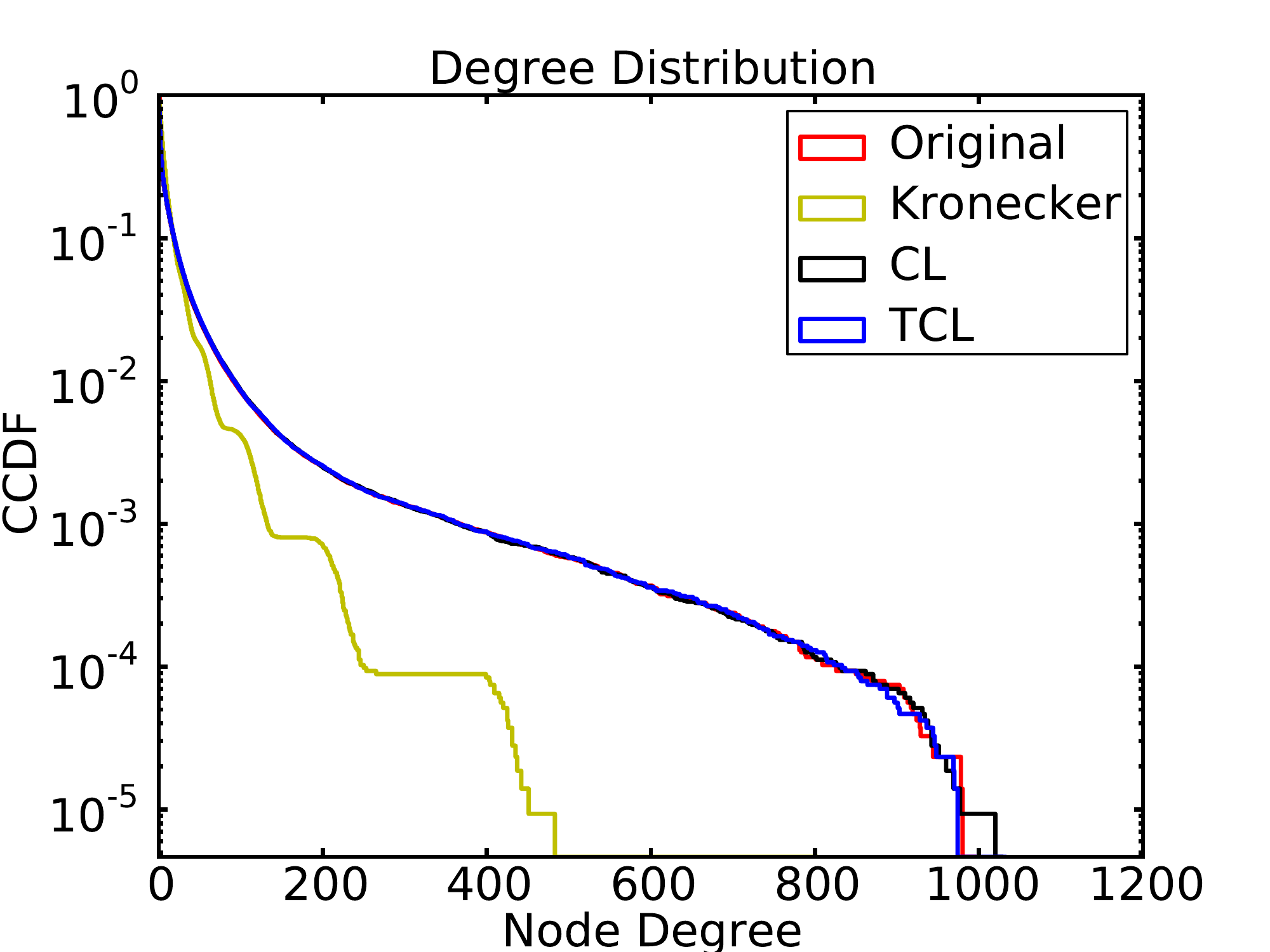}
\includegraphics[width=.33\textwidth]{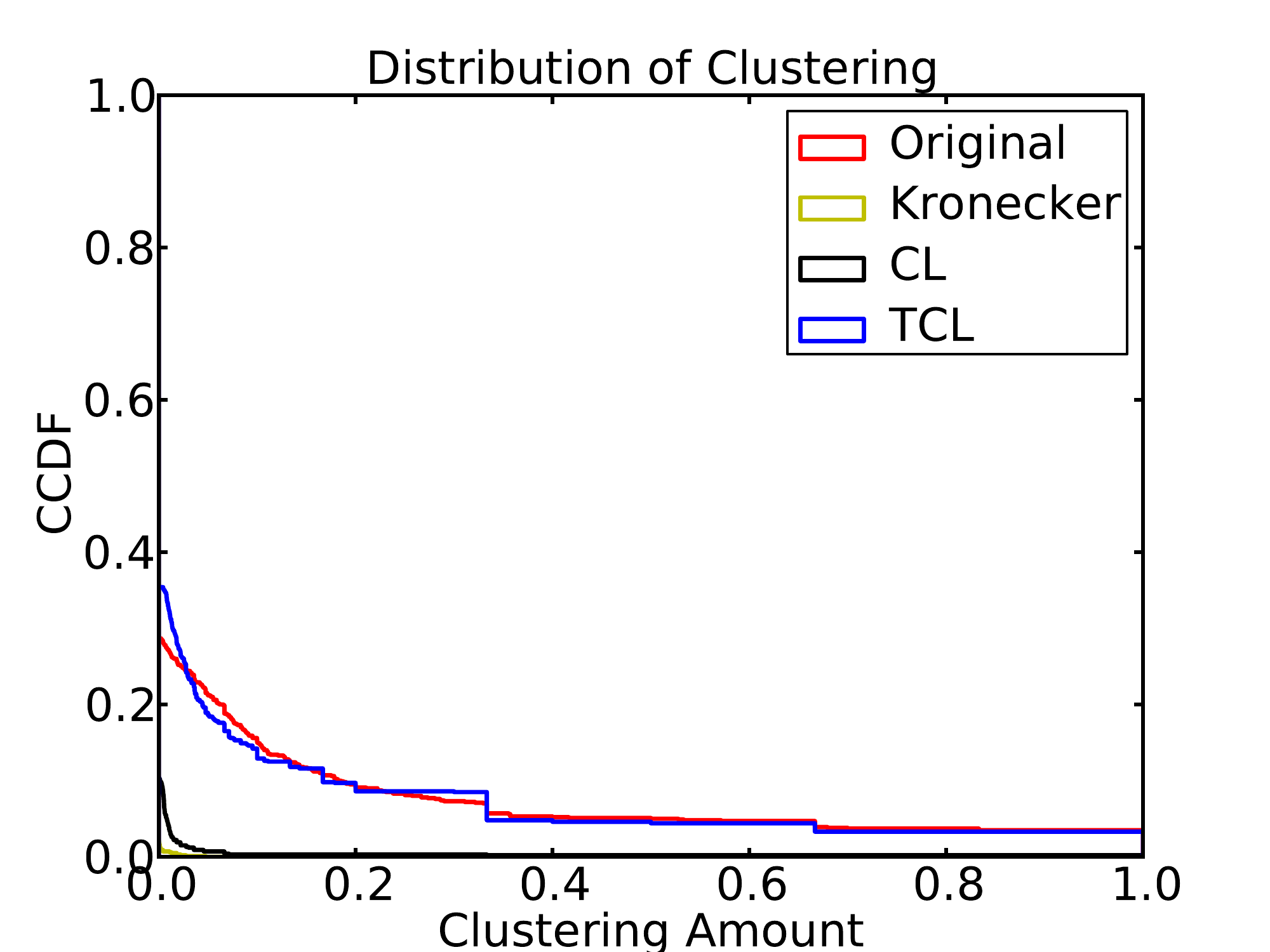}
\includegraphics[width=.33\textwidth]{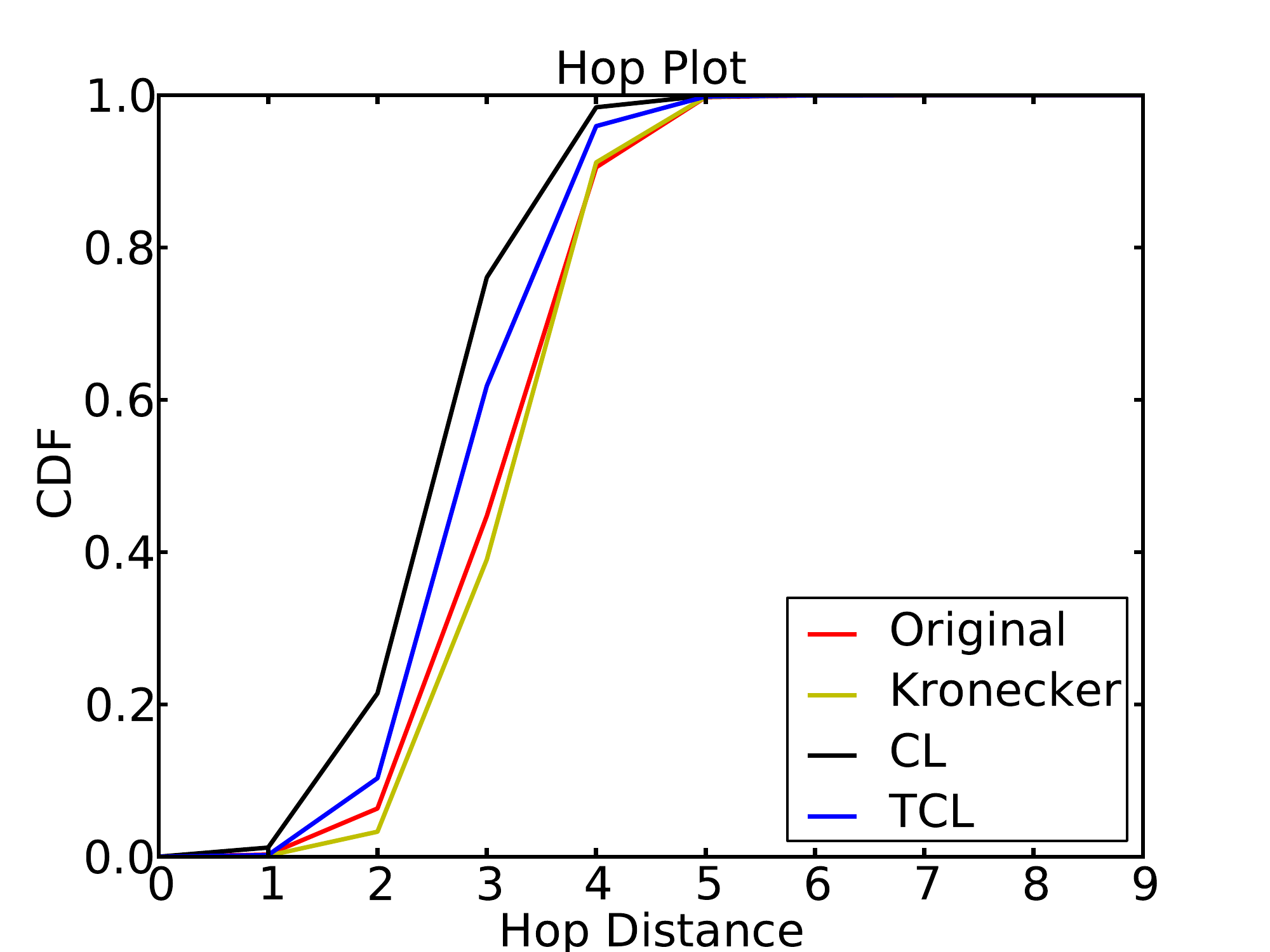}}

\caption{Degree distribution, clustering and hop plots for the Epinion, Facbook, Gnutella30 and PurdueEmail datasets.}
\label{fig:dd_cc_hp}
\end{figure*}

\subsection{Graph Statistics}
So far, the CL model has shown superiority in terms of learning and runtime to both TCL and KPGM, while TCL has distanced itself from KPGM in the same measures.  However, the ability to learn and generate large graphs quickly is only a portion of the task, as generating a network with little or no resemblance to the given network does not meet the primary goal of modeling the network.

In order to test the ability of the models to generate networks with similar characteristics to the original 4 networks, we compare them on three well known graph statistics: the degree distribution, the clustering coefficient and the hop plot.

Matching the degree distribution is the goal of both the CL and KPGM models, as well as the new TCL algorithm.  In the left hand column of Figure  \ref{fig:dd_cc_hp} ,the degree distributions of the networks generated from each model for each real-world is shown, compared against the original real-world networks' degree distribution.  The measure used along the y-axis is the complementary cumulative degree distribution (CCDF), while the x-axis plots the degree, meaning the y-value at a point indicates the percentage of nodes with greater degree.  The 4 networks have degree distributions of varying styles -- the 3 social networks (Epinions, Facebook, and PurdueEmail) have curved degree distributions, compared to Gnutella30 whose degree distribution is nearly straight, indicating an exponential cutoff. As theorized, both the CL and TCL have a degree distribution which closely matches their expected degree distribution, regardless of the distribution shape.  KPGM best matches the Gnutella30 network, sharing an exponential cutoff indicated by a straight line, but is still separated from the original network's distribution.  With the social networks KPGM has an alternating dip/flat line pattern which does not resemble the true degree distribution.  In contrast, TCL matches the distributions of all 4 networks with the same accuracy as the CL method, showing the model continues to match the degree distribution well even with the addition of transitive closures.

The next statistic we examine is TCL's ability to model clustering, as neither CL nor KPGM attempt to replicate the clustering found in social networks.  As with the degree, we plot the CCDF on the y-axis, but against the local clustering coefficient on the x-axis.  The clustering coefficient is a measure comparing the number of triangles in the network vs. the \emph{possible} number of triangles in the network, and a higher value indicates more clustering \cite{Watts_SmallWorld}.  On the network with the largest amount of clustering, Epinions, TCL matches the distribution of clustering coefficients well with the TCL distribution lying on top of the original distribution.  The same follows for Facebook and PurdueEmail, despite the large size of the latter.  The Gnutella30 has a remarkably low amount of clustering -- so low that it is plotted in log-log scale -- yet TCL is able to follow the distribution as well.  Furthermore, the networks exhibit a range of $\rho$ values, but the TCL EM estimation is able to accurately capture the clustering behavior of the original network.

In contrast, CL and KPGM cannot model the clustering distribution.  For each network, both methods lack appreciable amounts of clustering in their generated graphs, even undercutting the Gnutella30 network which has far less clustering than the others.  This shows a key weakness with both models, as clustering is an importation characteristic of small-world networks.

The last measure examined is the Hop Plot, in the right column of Figure \ref{fig:dd_cc_hp}.  The Hop Plot indicates how tightly connected the graph is; for each x-value, the y-value corresponds to the percentage of nodes that are reachable within that many hops.  When generating the hop plots, we excluded any nodes with infinite hop distance and discarded disconnected components and orphaned nodes.  All of the models followed the hop plots well, with TCL producing hop plots very close to the standard CL.  This indicates that the transitive closures of TCL did not impact the connectivity of the graph and the gains in terms of clustering can be obtained without altering the hop plot.


\section{Conclusions}\label{sec:conclusions}
In this paper we demonstrated a correction to the Chung Lu fast estimation algorithm and introduced the Transitive Chung Lu model.  Given a real-world network, the TCL model learns and generates a graph which accurately captures the degree distribution, clustering coefficient distribution and hop plot found in the training network.  We proved the algorithm generates a network in $O(M)$, on the order of CL and faster than KPGM.  The amount of clustering in the generated network is controlled by a single parameter, and we demonstrated how estimating the parameter is several orders of magnitude faster than estimating KPGM.  The networks generated by our TCL algorithm exhibit characteristics of the original network, including degree distribution and clustering, unlike the graphs generated by CL and KPGM.
Future directions for these results are numerous, including analysis of networks over time and methods which explore extrapolating a larger graph from a given graph.  Lastly, while our analysis has TCL generating networks which match the degree distributions and clustering of a real-world network, usage of a transitivity parameter for clustering is still a heuristic approach.  A more formal analysis of the clustering expected from such a model would be worth pursuing.

\bibliographystyle{abbrv}   
{\small
\bibliography{Pfeiffer_TCL_Arxiv}  
	}		     

\end{document}